\documentclass{sig-alternate}
\usepackage{amsmath,amssymb,url,amsfonts,mathrsfs}
\usepackage{graphics}
\usepackage{float}

\usepackage{color}

\def\old#1{}

\usepackage{enumerate}
\usepackage{verbatim}

\newcommand{\be}{\begin{equation}}
\newcommand{\ee}{\end{equation}}

\newcommand{\beps}{\varepsilon}

\newtheorem{theorem}{Theorem}[section]
\newtheorem{lemma}[theorem]{Lemma}

\newtheorem{corollary}[theorem]{Corollary}
\newtheorem{proposition}[theorem]{Proposition}
\newtheorem{definition}[theorem]{Definition}
\newtheorem{asmption}[theorem]{Assumption}

\newcommand{\cS}{\mathcal{S}}
\newcommand{\bsig}{\boldsymbol{\sigma}}
\newcommand{\blambda}{\boldsymbol{\lambda}}
\newcommand{\bLambda}{\boldsymbol{\Lambda}}
\newcommand{\bpi}{\boldsymbol{\pi}}
\newcommand{\bdel}{\boldsymbol{\delta}}

\newcommand{\Rp}{\mathbb{R}_{+}}

\newcommand{\Zp}{\mathbb{Z}_{+}}

\newcommand{\E}{\mathbb{E}}

\newcommand{\bQ}{\mathbf{Q}}

\newcommand{\bx}{\mathbf{x}}
\newcommand{\by}{\mathbf{y}}

\newcommand{\bone}{\mathbf{1}}

\providecommand{\boldsymbol}[1]{\mbox{\boldmath $#1$}}

\makeatother
\begin{document}

\title{Qualitative Properties of {\huge $\alpha$}-Weighted Scheduling Policies}
%
%
%
%
%

\numberofauthors{3} 
%
\author{
%
%
\alignauthor Devavrat Shah\\
       \affaddr{MIT, LIDS}\\
       \affaddr{Cambridge, MA 02139}\\
       \email{devavrat@mit.edu}\\
\alignauthor John N. Tsitsiklis\\
       \affaddr{MIT, LIDS}\\
       \affaddr{Cambridge, MA 02139}\\
       \email{jnt@mit.edu}\\
\alignauthor Yuan Zhong\\
       \affaddr{MIT, LIDS}\\
       \affaddr{Cambridge, MA 02139}\\
       \email{zhyu4118@mit.edu}\\
}

\maketitle
\begin{abstract}
We consider a switched network, a fairly general constrained  queueing network
model that has been used successfully to
model the detailed packet-level dynamics in communication networks, such as
input-queued switches and wireless networks. The main operational issue
in this model is that of deciding which queues to serve, subject to  certain constraints.
In this paper, we study qualitative performance properties of the well
known $\alpha$-weighted scheduling policies. The stability, in the  sense of positive
recurrence, of these policies has been well understood. We establish
exponential upper bounds on the tail of the steady-state distribution  of the backlog. 
Along the way, we prove finiteness of the
expected steady-state backlog when $\alpha<1$, a property that was  known only for
$\alpha\geq 1$.
Finally, we analyze the excursions of the maximum backlog over a finite
time horizon for $\alpha \geq 1$. As a consequence, for $\alpha \geq 1$,
we establish the full state space collapse property \cite{SW06, SW09}.
\end{abstract}

\section{Introduction}\label{sec:intro}

This paper studies various qualitative stability and performance properties
of the so-called $\alpha$-weighted policies, as applied to a switched  network model
(cf.\cite{TE92, SW09}).
This model is a special case of the ``stochastic processing network  model" (cf.\cite{Harrison}),
which has become the canonical framework for the study of
a large class of networked queueing systems, including systems arising
in communications, manufacturing, transportation, financial markets,
etc. 
The primary reason for the popularity of the switched network model
is its ability to faithfully model the behavior of a
broad spectrum of networks at a fine granularity. Specifically,
the switched network model is useful in describing packet-level  (``micro'')
behavior of medium access in a wireless network and of the input-queued switches that
reside inside Internet routers.
This model has proved tractable enough to allow for
substantial progress in understanding the stability and performance  properties
of various control policies.


At a high level, the switched network model involves a collection of queues. Work arrives to these queues exogenously
or from another queue and gets serviced; it then either leaves the  network
or gets re-routed to another queue. Service at the queues requires the  use of
some commonly shared constrained resources.
This leads to the problem of
{\em scheduling} the service of packets queued in the switched network.
To utilize the network resources efficiently, a properly designed  scheduling
policy is required.
Of particular interest are the popular Maximum Weight or
MW-$\alpha$ policies, introduced in \cite{TE92}.
They are the only known simple and universally applicable policies
with performance guarantees. In addition, the
MW-$\alpha$ policy has served as an important guide for designing
implementable algorithms for input-queued  switches and wireless
medium access (cf.\cite{islip,T98,GPS02,GNT06,RSS09}).
This motivates the work in this paper, which focuses on certain
qualitative properties of
MW-$\alpha$ policies.

\paragraph{Related Prior Work} Because of the significance of the
$\alpha$-weighted policies,
there is a large body of research on their properties. We provide here
 a brief overview of the work that is most relevant to our purposes.

The most basic performance question concerns 
throughput and stability. 
Formally, we say that an algorithm is \emph{throughput optimal} or  \emph{stable} if
the underlying network Markov chain is positive
recurrent whenever the system is {\em underloaded}.
For the MW-$\alpha$ policy, under a general enough stochastic
model, stability has been established 
for any $\alpha > 0$ (cf.\ \cite{TE92,MAW96,DP00,AKRSVW04}).


A second, finer, performance question
concerns the evaluation
of the average backlog in the system, in steady-state.
Bounds on the average backlog are usually obtained by considering the  same stochastic
Lyapunov function that was used to prove stability,
and by building on the drift inequalities established in the course of  the stability
proof; see, e.g., \cite{DM95}. Using this approach, it  is known that the average
expected backlog under $\alpha$-weighted  policies is finite, when $\alpha\geq 1$ (\cite{K-M}).
However, such a result is not
known when $\alpha \in (0,1)$.

An important performance analysis method that
has emerged over the past few decades focuses on the {\em heavy  traffic} regime, in
which the
system is
loaded near capacity. For the switched model,
heavy traffic analysis has revealed some
intriguing relations between the policy parameter $\alpha$
and the performance of the system through a phenomenon known as {\em  state
space collapse}. In particular, in the heavy traffic limit and for an  appropriately
scaled version of the system, the state evolves in a much lower-dimensional space
(the state space ``collapses''). The structure of the collapsed state
space provides important information about the system
behavior (cf.\ \cite{KW04,SW06,SW09}). Under certain
somewhat specific assumptions, a complete heavy traffic analysis
of the switched network model
has been
carried out in \cite{S04, DL08}.
However, for the more general switched network model, only a weaker  result
is available, involving a so-called multiplicative state space
collapse property \cite{SW09,SW06}. 
State space collapse results are related to
understanding certain transient properties of the
network, such as the evolution of the queues over a finite time horizon.
To the best of our knowledge, a transient analysis of the switched  network
model is not available.

A somewhat different approach focuses on tail probabilities of the
steady-state backlog and the associated large deviation principle  (LDP). This approach
provides important
insights  about the overflow probability in the presence of finite   buffers.
There have been notable works in this direction, for specific  instances of the
switched network model, e.g., \cite{S-LDP08}. In a similar setting,  the reference
 \cite{KJS08}
has also established a LDP for the MW-1 policy, using  Garcia's
extended-contraction principle for quasi-continuous mappings.
More recently, \cite{VL07, VL09} has announced a characterization of  the precise tail
behavior of the $(1+\alpha)$-norm
of the backlog, under the MW-$\alpha$ policy.  However, in these works, the LDP exponent is only given  implicitly, as the solution of a
complicated, possibly infinite dimensional optimization problem.



\paragraph{Our Contributions}

We establish various qualitative performance bounds
for $\alpha$-weighted policies, under the switched network model.
In the stationary regime, we establish finiteness of the expected backlog, and an
exponential upper bound on the steady-state tail probabilities of the  backlog. In the
transient regime, we establish a maximal inequality
 on the queue-size process, and the strong state space
collapse property under $\alpha$-weighted
policies, when $\alpha\geq 1$. Our analysis is based on drift  inequalities on
suitable Lyapunov functions. Our methods, however,  depart from prior work because they
rely on different classes of  Lyapunov functions, and also involve some new techniques.

In more detail, we begin by establishing the finiteness of the steady-state expected
backlog under the MW-$\alpha$
policy, for any $\alpha \in (0,1)$. Instead of the traditional  Lyapunov function
$\|\cdot\|_{\alpha+1}^{\alpha+1}$, we rely on a  Lyapunov function which is a suitably
smoothed version of $\|\cdot\| _{\alpha+1}^2$.

We continue by deriving a drift inequality for a ``norm'' or  ``norm-like''
Lyapunov function, namely, $\|\cdot\|_{\alpha+1}$  or a suitably  smoothed version.
Using the drift inequality,
we establish exponential
tail bounds for the steady-state backlog distribution under the MW-$ \alpha$
policy, for  any $\alpha \in (0,\infty)$.
Our method builds on certain results from \cite{BGT01} that allow  us to translate
drift inequalities into closed-form tail bounds; it  yields an explicit bound on the
tail exponent, in terms
of the system load and the total number of queues. This
is in contrast with the earlier work in \cite{S-LDP08,VL09}. That work  provides an exact
but implicit characterization of the tail exponents,
in terms of a complicated optimization problem, and provides
no immediate insights on the dependence of the tail exponents on the system  parameters,
such as the load and the number of queues. Furthermore, in contrast to  the
sophisticated mathematical techniques used in \cite{S-LDP08,VL09},
our explicit bounds are obtained
through elementary methods. For some additional perspective, we  also consider a
special case and compare our upper bound with  available lower bounds.

Finally, we provide a transient analysis under MW-$\alpha$ policies,  for the case where
$\alpha\geq 1$.
We use a Lyapunov drift inequality to obtain a bound on the  probability that the maximal
backlog over a given finite time interval  exceeds
a certain threshold.
This bound leads to the resolution
of the strong state space collapse conjecture
for the switched network model when $\alpha \geq 1$. This strengthens  the multiplicative
state space collapse results in \cite{SW06,SW09}.

\paragraph{Organization of the Paper}

The rest of the paper is organized as follows. In Section 2, we define the notation we
will employ, and describe the switched network model.  In Section 3, we provide formal
statements of our main results. In  Section 4, we establish a drift inequality for a
suitable Lyapunov  function, which will be key to the proof of the exponential upper 
bound on tail probabilities. In Section 5, we prove the finiteness of  steady-state
expected backlog when $\alpha \in (0,1)$. We prove the exponential upper bound in
Section 6. For a special instance, we compare this upper bound with available lower bounds in the Appendix.
The transient analysis is presented in Section 7. We start with a general
lemma, and specialize  it to obtain a maximal inequality under the MW-$\alpha$ policy,
for $ \alpha\geq 1$. We then apply the latter inequality to prove the full  state space
collapse result for $\alpha\geq 1$. We conclude the paper  with a brief discussion in
Section 8.  

\section{Model and Notation}
\subsection{Notation}
We introduce here the notation
that will be employed throughout the paper.
We denote the real vector space of dimension $M$ by $\mathbb{R}^M$ and  the set of
nonnegative $M$-tuples by $\mathbb{R}_{+}^M$. We write $\mathbb{R}$  for $\mathbb{R}^1$,
and $\mathbb{R}_{+}$ for $\mathbb{R}_+^1$.
We let $\mathbb{Z}$ be the set of integers,
$\mathbb{Z}_+$ the set of nonnegative integers, and $\mathbb{N}$ the  set of positive
integers.

For any vector $\mathbf{x}\in\mathbb{R}^M$, and any $\alpha>0$,  we define
$$\| \mathbf{x}\|_{\alpha}= \left( \sum_{i=1}^M |x_i|^{\alpha}
\right)^{1/\alpha}.$$
For any two vectors $\mathbf{x}=(x_i)_{i=1}^M$
and $\mathbf{y}=(y_i)_{i=1}^M$ of the same dimension, we let $ \mathbf{x}\cdot\mathbf{y}=
\sum_{i=1}^M x_i y_i$ be the dot product of  $\mathbf{x}$ and $\mathbf{y}$. For two real
numbers $x$ and $y$, we let
$x\vee y = \max\{x,y\}$. We also let $[x]^+ = x\vee 0$.
 We introduce the Kronecker delta symbol
$\delta_{ij}$, defined as $\delta_{ij}=1$ if $i=j$, and $ \delta_{ij}=0$ if $i\neq j$.
We let $\mathbf{e}_i = (\delta_{ij})_{j=1}^M$ be the $i$-th unit  vector in
$\mathbb{R}^M$, and $\bone$ the vector of all ones. For a set $S$, we  denote its
cardinality by $|S|$, and its indicator function by $ \mathbb{I}_S$. For a matrix $A$, we
let $A^T$ denote its transpose.
We will also use the abbreviations ``RHS/LHS'' for ``right/left-hand  side,'' and ``iff''
for ``if and only if.''
\subsection{Switched Network Model}\label{ssec:switch}
\paragraph{The Model}
We adopt the model in \cite{SW09}, while restricting to the case of  {\em single-hop}
networks, for ease of exposition. However, our results naturally  extend to
{\em multi-hop} models, under the ``back-pressure'' variant of the MW-$ \alpha$ policy.

Consider a collection of
$M$ queues. Let time be discrete: timeslot $\tau \in \{0,1,\ldots \}$
runs from time $\tau$ to $\tau+1$. Let $Q_i(\tau)$  denote the  (nonnegative integer)
length of queue $i\in \{1,2,\ldots,M \}$ at the  beginning of timeslot $\tau$, and
let $\mathbf{Q}(\tau)$ be the vector $(Q_i(\tau))_{i=1}^{M}$. Let
$\mathbf{Q}(0)$ be the vector of initial queue lengths.

During each timeslot $\tau$, the queue vector $\mathbf{Q}(\tau)$
is offered service described by a vector $\boldsymbol{\sigma} (\tau)=(\sigma_i
(\tau))_{i=1}^{M}$ drawn from a given
finite set $\mathcal{S} \subset \{0,1\}^M$ of \emph{feasible schedules}.
Each queue
$i\in \{1,2,\ldots,M\}$ has a dedicated exogenous arrival process
$(A_i(\tau))_{\tau\geq 0}$, where $A_i(\tau)$ denotes the number of packets 
that arrive to queue $i$ up to the beginning of timeslot $\tau$,
and $A_i(0)=0$ for all $i$. We also let $a_i(\tau)=A_i(\tau+1)- A_i(\tau)$, which is
the number of packets 
that arrive to queue $i$ during timeslot
$\tau$. For simplicity, we assume  that the $a_i( \cdot)$
are independent Bernoulli processes with parameter $\lambda_i$.
We call $\boldsymbol{\lambda} = (\lambda_i)_{i=1}^M$
the \emph{arrival rate vector}.

Given the service schedule $\boldsymbol{\sigma}(\tau) \in \mathcal{S}$  chosen
at timeslot $\tau$, the queues evolve according to the relation \begin{displaymath}
Q_i(\tau+1)=\big [Q_i(\tau)-\sigma_i(\tau)\big ]^{+} + a_i(\tau).
\end{displaymath}
In order to avoid trivialities, we assume,  throughout the paper, the  following.
\begin{asmption}\label{as1}
For every queue $i$, there exists a $\boldsymbol{\sigma}\in\mathcal{S} $ such that
$\sigma_i = 1$.
\end{asmption}
\paragraph{An example: Input Queued (IQ) Switches} The switched network
model captures important instances of communication network scenarios
(see \cite{SW09} for various examples). Specifically, it faithfully
models the  packet-level operation of an input-queued (IQ) switch
inside an Internet router. For an $m$-port IQ switch, it
has $m$ input and $m$ output ports. It has a separate queue for each
input-output pair $(i,j)$, denoted by  $Q_{ij}$,\footnote{Here
we deviate from our convention of
indexing queues by a single subscript. This will ease exposition in
the context of IQ switches, without causing confusion.} for a  total of
$M = m^2$ queues. A schedule is required to match each input to  exactly one
output, and each output to exactly one input. Therefore,
the set of schedules $\cS$ is
$$\left\{ \bsig = (\sigma_{ij}) \in \{0,1\}^{m \times m} : \
\sum_{k=1}^m \sigma_{ik} = \sum_{k=1}^m \sigma_{kj} = 1, ~\forall  ~i, j \right\}. $$
We assume that 
the arrival process at each queue $Q_{ij}$ is an independent
Bernoulli process with mean $\lambda_{ij}$.

\paragraph{Capacity Region} We define the capacity region $\bLambda$
of a switched network as
$$\left\{ \blambda \in \Rp^{M} :
\blambda \leq \sum_{\bsig \in \cS} \alpha_{\bsig} \bsig,
 \alpha_{\bsig} \geq 0, ~\forall~\bsig \in \cS,
\sum_{\bsig \in \cS} \alpha_{\bsig} < 1\right\}.$$
It is called the capacity region because there exists a policy for which
the Markov chain describing the network is positive recurrent iff $\blambda \in  \bLambda$. We define
the {\em load}
induced by $\blambda \in \bLambda$, denoted by $\rho(\blambda)$, as
$$ \rho(\blambda) ~=~ \inf \left\{\sum_{\bsig \in \cS} \alpha_{\bsig}  ~:~
\blambda \leq \sum_{\bsig \in \cS} \alpha_{\bsig} \bsig, \quad  \alpha_{\bsig} \geq 0,
~\forall~\bsig \in \cS \right\}. $$
Note that $\rho(\blambda)<1$, for all $\blambda \in \bLambda$.

\paragraph{The Maximum-Weight-{\large $\alpha$} Policy} We now describe the so-called \emph{Maximum-Weight-$\alpha$} (MW-$\alpha$) policy. For $\alpha > 0$, we use $\mathbf{Q} (\tau)^{\alpha}$ to denote the
vector
$(Q_i^{\alpha}(\tau))_{i=1}^{M}$. We define the \emph{weight} of  schedule
$\boldsymbol{\sigma} \in \mathcal{S}$ to
be $\boldsymbol{\sigma \cdot Q}(\tau)^{\alpha}$. The MW-$\alpha$  policy chooses, at each
timeslot $\tau$, a schedule
with the largest weight (breaking ties arbitrarily). Formally,
during timeslot $\tau$, the policy chooses a schedule $ \boldsymbol{\sigma}(\tau)$
that
satisfies
\begin{displaymath}
\boldsymbol{\sigma}(\tau)\cdot \boldsymbol{Q}(\tau)^{\alpha} = 
\max_{\boldsymbol{{\sigma}}\in \mathcal{S}}\boldsymbol{{\sigma}\cdot Q} (\tau)^{\alpha}.
\end{displaymath}
We define the maximum $\alpha$-weight of the queue length vector $ \mathbf{Q}$ by
$w_{\alpha}(\mathbf{Q})=\max_{\boldsymbol{\sigma} \in
\mathcal{S}}\boldsymbol{\sigma\cdot Q}^{\alpha}$. When $ \alpha=1$, the policy is
simply called the MW policy, and we use the  notation $w(\mathbf{Q})$ instead of
$w_1(\mathbf{Q})$. We take note of
the fact that under the MW-$\alpha$ policy, the resulting Markov chain  is known to be
positive recurrent,
for any $\blambda \in \bLambda$ (cf. \cite{MAW96}).

\section{Summary of Results}
In this section,  we summarize our main results 
for both the steady-state and the
transient regime. The proofs are given in subsequent sections.

\subsection{Stationary regime}

The Markov chain $\mathbf{Q}(\cdot)$ that describes a
switched network operating under the MW-$\alpha$ policy is known to be
positive recurrent, as long as the system is underloaded, i.e., if
$\blambda \in \bLambda$ or, equivalently, $\rho(\blambda)<1$.
It is not hard to verify that this
Markov chain is irreducible and aperiodic. Therefore, there exists a  unique stationary
distribution, which we will denote by $\bpi$. We use $\E_{\bpi}$ and $\mathbb{P}_{\bpi}$ 
to denote expectations and probabilities under $\bpi$.

\paragraph{Finiteness of Expected Queue-Size} We establish that under  the MW-$\alpha$
policy, the steady-state expected queue-size
is finite, for any
$\alpha \in (0,1)$. (Recall that this result is already known  when  $\alpha\geq
1$.)
\begin{theorem}\label{THM:SW1}
Consider a switched network operating under the MW-$\alpha$ policy
with $\alpha \in (0,1)$, and assume that $\rho(\blambda)<1$. 
Then, the steady-state expected  queue-size
is finite, i.e.,
$$ \E_{\bpi}\left[\|\bQ\|_1\right] < \infty.$$
\end{theorem}

\paragraph{Exponential Bound on Tail Probabilities} For the MW-$\alpha $ policy, and for
any $\alpha \in (0,\infty)$, we obtain
an explicit exponential upper bound on the tail probabilities   of the queue-size,
in steady-state.
Our result involves two constants defined by
$$\bar{\nu} = \mathbb{E}\big [ \|\mathbf{a}(1)\|_{\alpha+1}\big ],
\qquad\qquad \gamma = \frac{1-\rho}{2M^{\frac{\alpha}{\alpha+1}}},$$
where $\rho=\rho(\blambda)$.
\begin{theorem}\label{thm:sw2}
Consider a switched network operating under the MW-$\alpha$ policy,
and assume that $\rho=\rho(\blambda)<1$.  
There exist positive constants $B$ and $B'$ such that for all $\ell \in \Zp$:
\begin{itemize}
\item[(a)] if $\alpha\geq 1$, then
\begin{displaymath}
\mathbb{P}_{\bpi}\left( \|\mathbf{Q}(\tau)\|_{\alpha+1}> B + 2 
M^{\frac{1}{\alpha+1}}\ell\right)\leq \left(\frac{\bar{\nu}}{\bar{\nu}+
\gamma}\right)^{\ell+1};
\end{displaymath}

\item[(b)]
if $\alpha \in (0,1)$, then
\begin{displaymath}
\mathbb{P}_{\bpi}\left( \|\mathbf{Q}(\tau)\|_{\alpha+1}> B' + 10 
M^{\frac{1}{\alpha+1}}\ell\right)\leq \left(\frac{5 \bar{\nu}} {5\bar{\nu}
+\gamma}\right)^{\ell+1}.
\end{displaymath}
\end{itemize}
\end{theorem}

Note that Theorem \ref{THM:SW1} could be obtained as a simple  corollary of Theorem
\ref{thm:sw2}. On the other hand, our proof of  Theorem \ref{thm:sw2} requires the
finiteness of $\E_{\bpi}\left[\|\bQ \|_1\right]$, and so Theorem \ref{THM:SW1} needs to
be established  first.

\old{Theorem \ref{thm:sw2} yields a simple bound on the tail exponent.  We can further
simplify this bound to make the dependence on the  system parameters more explicit. We
concentrate on the heavily-loaded  case, where $\rho$ increases to 1, so that $\gamma$
tends to zero. We  have
\begin{align}
& \lim\sup_{R \to\infty} \frac{1}{R} \log \mathbb{P}_{\bpi}\big( \|
\mathbf{Q}(\tau)\|_{\alpha+1} > R\big)\nonumber \\
& \leq -\frac{M^{-\frac{1}{1+\alpha}}}{10} \log \left(1 +  \frac{\gamma}{5\bar{\nu}}
\right) \stackrel{\rho \to 1}{\approx} -\frac{M^{-\frac{1}{1+\alpha}}}{50} 
\frac{\gamma}{\bar{\nu}}\nonumber \\
& = -\frac{1-\rho}{100 M \bar{\nu}} \leq -\frac{1-\rho}{100 M \|\blambda\|^{\frac{1}{1+\alpha}}_1}. \label{eq:sw2-1}
\end{align}
For the last inequality, we use the fact that the
arrival process is Bernoulli and argue as follows:
\begin{align*}
\bar{\nu} & = \E\big [ \|\mathbf{a}(1)\|_{\alpha+1}\big ] = \E\left[\left(\sum_i a_i(1)\right)^{1/(1+\alpha)} \right] \\
          & \leq \left(\E\left[\sum_i a_i(1)\right]\right)^{1/(1+ \alpha)} \\
          & = \left(\sum_i \lambda_i\right)^{1/(1+\alpha)} ~=~ \|\blambda\|_1^{\frac{1}{1+\alpha}}.
\end{align*}}
In the Appendix, we
comment on the tightness of our upper bounds by comparing them with
explicit lower bounds that follow from the recent
large deviations results in \cite{VL09}.

\subsection{Transient regime}
Here we provide a simple inequality on the maximal excursion of the  queue-size over a
finite time interval, under the MW-$\alpha$ policy, with  $\alpha \geq 1$.
\begin{theorem}\label{thm:max}
Consider a switched network operating under the MW-$\alpha$ policy
with $\alpha \geq 1$, and assume that $\rho(\blambda)<1$.
Suppose that $\bQ(0)=\boldsymbol{0}$. Let $Q_{\max}(\tau) = \max_{i\in  \{1,\ldots,M\}}
Q_i(\tau)$,
and $Q^*_{\max}(T) = \max_{\tau\in\{0,1,\ldots,T\}} Q_{\max}(\tau)$. Then, for any $b>0$,
\begin{equation}
\mathbb{P}\left(Q^*_{\max}(T)\geq b\right)\leq \frac{K(\alpha,M) T}{(1-
\rho)^{\alpha-1}b^{\alpha+1}},
\end{equation}
for some positive constant $K(\alpha,M)$ depending only on $\alpha$  and $M$.
\end{theorem}
As an important application, we use Theorem \ref{thm:max} to prove a full state space collapse result,\footnote{
This is strong state space collapse and not full diffusion approximation.}
for $\alpha\geq 1$, in Section \ref{ssec:ssc}. The precise statement can be found in Theorem \ref{thm:ssc}.

\section{MW-{\Large $\alpha$} policies: A Useful Drift Inequality}\label{sec:drift}

The key to many of our results 
is a {\em drift inequality} that holds
for every $\alpha>0$ and $\blambda \in \bLambda$. 
In this section, 
we shall state and prove this inequality. It will be used in Section \ref{sec:exp} to 
prove Theorem \ref{thm:sw2}. We remark that 
similar drift inequalities, but for a different Lyapunov function, 
have played an important role in establishing positive recurrence
(cf.\ \cite{TE92}) and multiplicative state space collapse (cf.\  \cite{SW09}).

We will be making extensive use of 
a second-order mean value
theorem \cite{nlp}, which we state below for easy reference. 

\begin{proposition}\label{prp:mvt}
Let $g : \mathbb{R}^M \rightarrow \mathbb{R}$ be twice continuously
differentiable over an open sphere $S$ centered at a vector $\mathbf{x}$.
Then, for any $\mathbf{y}$ such that $\mathbf{x}+\mathbf{y} \in S$, there
exists a $\theta \in [0,1]$ such that
\begin{equation}\label{eq:taylor}
g(\mathbf{x}+\mathbf{y})=g(\mathbf{x})+\mathbf{y}^T\nabla g(\mathbf{x})+
\frac{1}{2}\mathbf{y}^T H(\mathbf{x}+\theta\mathbf{y}) \mathbf{y},
\end{equation}
where $\nabla g(\mathbf{x})$ is the gradient of $g$ at $\mathbf{x}$, and
$H(\mathbf{x})$ is the Hessian of the function $g$ at $\mathbf{x}$.
\end{proposition}

We now define the Lyapunov function that we will employ. 
For $\alpha \geq 1$, it will be 
simply the $(\alpha+1)$-norm $\|\mathbf{x}\|_{1+\alpha}$ of a vector $\mathbf{x}$. However, when $\alpha\in(0,1)$, this function has unbounded second derivatives as we approach the boundary of $\mathbb{R}_+^M$. For this reason, our Lyapunov function will be a suitably smoothed version of $\| \cdot \|_{\alpha+1}$.

\begin{definition}\label{df:sw-lyap}
Define $f_{\alpha} : \mathbb{R}_+ \rightarrow \mathbb{R}_+$
to be $f_{\alpha}(r) = r^{\alpha}$, when $\alpha \geq 1$, and
\begin{displaymath}
f_{\alpha}(r) = \left\{ \begin{array}{ll}
r^{\alpha}, & \textrm{ if } r\geq 1, \\
(\alpha-1) r^3 +(1-\alpha)r^2 + r, & \textrm{ if } r \leq 1,
\end{array} \right.
\end{displaymath}
when $\alpha \in (0,1)$. Let $F_{\alpha} : \mathbb{R}_+ \rightarrow \mathbb{R}_+$
be the antiderivative of $f_{\alpha}$, so that $F_{\alpha}(r) = \int_{0}^{r} f_{\alpha}(s)\, ds$.
The Lyapunov function $L_{\alpha} : \mathbb{R}_+^M \rightarrow \mathbb{R}_+$ is defined
to be
\begin{displaymath}
L_{\alpha} (\mathbf{x}) = \left [ (\alpha+1)\sum_{i=1}^M F_{\alpha}(x_i)\right ]^{\frac{1}{\alpha+1}}.
\end{displaymath}
\end{definition}

We will make heavy use of various properties
of the functions $f_{\alpha}$, $F_{\alpha}$, and $L_{\alpha}$, which we summarize in the following lemma. The proof is elementary and is omitted.

\begin{lemma}\label{lm:prpt}
Let $\alpha \in (0,1)$. The function $f_{\alpha}$ has the following properties:
\begin{itemize}
\item[(i)] it is continuously differentiable with
$f_{\alpha}(0)=0,\ f_{\alpha}(1)=1,\ f'_{\alpha}(0)=1$, and $f'_{\alpha}(1) = \alpha$;
\item[(ii)] it is increasing and, in particular, $f_{\alpha}(r)\geq 0$ for all $r\geq 0$;
\item[(iii)] we have $r^{\alpha}-1\leq f_{\alpha}(r)\leq r^{\alpha}+1$, for all $r\in [0,1]$;
\item[(iv)] $f'_{\alpha}(r)\leq 2$, for all $r\geq 0$.
\end{itemize}
Furthermore, from (iii), we also have the following property of $F_{\alpha}$:
\begin{itemize}
\item[(iii')] $r^{\alpha+1}-2\leq (\alpha+1)F_{\alpha}(r)\leq r^{\alpha+1}+2$ for all $r\geq 0$. \end{itemize}
\end{lemma}

We are now ready to state the drift inequality.
\begin{theorem}\label{thm:sw-drift}
Consider a switched network operating under the MW-$\alpha$ policy, and assume that 
$\rho=\rho(\blambda) < 1$. Then, there exists a constant $B>0$, such that
if $L_{\alpha}(\mathbf{Q}(\tau))>B$, then
\begin{equation}\label{eq:sw-drift}
\mathbb{E}[L_{\alpha}(\mathbf{Q}(\tau+1)) - L_{\alpha}(\mathbf{Q}(\tau)) \mid \mathbf{Q}(\tau)]
\leq -\frac{1-\rho}{2}M^{\frac{1}{\alpha+1}-1}.
\end{equation}
\end{theorem}
The proof of this drift inequality is quite tedious when $\alpha \neq 1$. To make the proof more accessible and to provide 
intuition, we first present the somewhat simpler proof for
$\alpha=1$. We then provide the proof for the case of general $\alpha$, by considering  separately the two 
cases where $\alpha > 1$ and $\alpha \in (0,1)$.

We wish to draw attention here to the main difference from related drift inequalities in the literature. The usual proof of stability involves the Lyapunov function $\| \mathbf{Q}\|_{\alpha+1}^{\alpha+1}$; for instance, for the standard MW policy, it involves a quadratic Lyapunov function. In contrast, we use $\| \mathbf{Q}\|_{\alpha+1}$
(or its smoothed version), which scales linearly along radial directions. In this sense, our approach is similar in spirit to 
\cite{BGT01}, which employed piecewise linear Lyapunov functions to derive drift inequalities and then moment and tail bounds.

\subsection{Proof of Theorem \ref{thm:sw-drift}: {\large $\alpha = 1$}}
In this section, we assume that 
$\alpha = 1$. As remarked earlier, we have $L_{\alpha}(\mathbf{x}) = \|\mathbf{x}\|_2$. 

Suppose that 
$\|\mathbf{Q}(\tau)\|_2>0$. We claim that on every sample path, we have 
\begin{equation}\label{eq:d1}
\|\mathbf{Q}(\tau+1)\|_2 - \|\mathbf{Q}(\tau)\|_2\leq \frac{\boldsymbol{Q}(\tau)\cdot\boldsymbol{\delta}(\tau)+\|\boldsymbol{\delta}(\tau)\|_2^2}{\|\mathbf{Q}(\tau)\|_2},
\end{equation}
where $\bdel(\tau) = \bQ(\tau+1)-\bQ(\tau)$.
To see this, we proceed as follows. We have
\begin{align}
& \left(\|\mathbf{Q}(\tau)\|_2+\frac{\mathbf{Q}(\tau)\cdot \boldsymbol{\delta}(\tau)+\|\boldsymbol{\delta}(\tau)\|_2^2}{\|\mathbf{Q}(\tau)\|_2}\right)^2 \nonumber \\
&\geq  \|\mathbf{Q}(\tau)\|_2^2+2\left(\boldsymbol{Q}(\tau)\cdot\boldsymbol{\delta}(\tau)+\|\boldsymbol{\delta}(\tau)\|_2^2\right)\nonumber \\
&\geq \|\mathbf{Q}(\tau)\|_2^2+2\boldsymbol{Q}(\tau)\cdot\boldsymbol{\delta}(\tau)+\|\boldsymbol{\delta}(\tau)\|_2^2 \nonumber \\
& = \|\mathbf{Q}(\tau)+\boldsymbol{\delta}(\tau)\|_2^2 = \|\mathbf{Q}(\tau+1)\|_2^2 \label{eq:jt1}.
\end{align}
Note that
\begin{align*}
& \|\mathbf{Q}(\tau)\|^2_2+\boldsymbol{Q}(\tau)\cdot\boldsymbol{\delta}(\tau)+\|\boldsymbol{\delta}(\tau)\|_2^2 \\
&= \left\|\mathbf{Q}(\tau)+\frac{\boldsymbol{\delta}(\tau)}{2}\right\|_2^2+\frac{3}{4}\|\boldsymbol{\delta}(\tau)\|_2^2\geq 0.
\end{align*}
We divide by $\|\mathbf{Q}(\tau)\|_2$, to obtain
\begin{displaymath}
\|\mathbf{Q}(\tau)\|_2+\frac{\boldsymbol{Q}(\tau)\cdot\boldsymbol{\delta}(\tau)+\|\boldsymbol{\delta}(\tau)\|_2^2}{\|\mathbf{Q}(\tau)\|_2}\geq 0.
\end{displaymath}
Therefore, we can take square roots of both sides 
of \eqref{eq:jt1}, without reversing the direction of the inequality, and the claimed
inequality \eqref{eq:d1} follows. 

Recall that $|\delta_i(\tau)|\leq 1$, because of the Bernoulli arrival assumption. It follows that $\|\boldsymbol{\delta}(\tau)\|_2\leq M^{1/2}$.
We now take the conditional expectation of both sides of \eqref{eq:d1}.
We have  
\begin{align}
& \mathbb{E}\left[\|\mathbf{Q}(\tau+1)\|_2 - \|\mathbf{Q}(\tau)\|_2 ~\Big|~ \mathbf{Q}(\tau)\right] \nonumber \\
&\leq \mathbb{E}\left[\frac{\boldsymbol{Q}(\tau)\cdot \mathbf{a}(\tau)-\mathbf{Q}(\tau)\cdot\boldsymbol{\sigma}(\tau)+M}{\|\mathbf{Q}(\tau)\|_2} ~\Big|~ \mathbf{Q}(\tau)\right]\nonumber \\
&= \frac{\sum_{i=1}^M Q_i(\tau)\mathbb{E}\left[a_i(\tau)\right]-\mathbf{Q}(\tau)\cdot \boldsymbol{\sigma}(\tau)+M}{\|\mathbf{Q}(\tau)\|_2}\nonumber \\
&= \frac{\sum_{i=1}^M Q_i(\tau)\lambda_i - w(\mathbf{Q}(\tau))+M}{\|\mathbf{Q}(\tau)\|_2}\nonumber \\
&\leq \frac{M-(1-\rho)w(\mathbf{Q}(\tau))}{\|\mathbf{Q}(\tau)\|_2}. \label{eq:d2}
\end{align}
The last inequality above is justified as follows.
From the definition of $\rho=\rho(\blambda)$, there exist constants
$\alpha_{\bsig} \geq 0$ such that $\sum_{\bsig\in \cS} \alpha_{\bsig} \leq \rho$, and 
\begin{eqnarray} 
\blambda  & \leq & \sum_{\bsig\in\cS} \alpha_{\bsig} \bsig.\label{eq:lamb}
\end{eqnarray}
Therefore, 
\begin{align}
\sum_i Q_i(\tau) \lambda_i & = \bQ(\tau) \cdot \blambda \leq \sum_{\bsig\in\cS} \alpha_{\bsig} \bQ(\tau)\cdot \bsig \nonumber \\
                           & \leq \sum_{\bsig\in\cS} \alpha_{\bsig} w(\bQ(\tau)) \leq \rho w(\bQ(\tau)).\label{eq:lamb-drift}
\end{align}
Let $Q_{\max}(\tau) = \max_{i=1}^M Q_i(\tau)$. Then, 
\begin{displaymath}
\|\mathbf{Q}(\tau)\|_2 \leq (M Q^2_{\max}(\tau))^{\frac{1}{2}} = M^{\frac{1}{2}}Q_{\max}(\tau).
\end{displaymath}
From Assumption \ref{as1}, we have
\begin{displaymath}
w(\mathbf{Q}(\tau))\geq Q_{\max}(\tau).
\end{displaymath}
Therefore, the RHS of \eqref{eq:d2} can be upper bounded by 
$$ -(1-\rho)M^{-1/2} + \frac{M}{\|\mathbf{Q}(\tau)\|_2} ~\leq~ -\frac{1}{2}(1-\rho)M^{-1/2}, $$ 
when $\|\mathbf{Q}(\tau)\|_2$ is sufficiently large.

\subsection{Proof of Theorem \ref{thm:sw-drift}: {\large $\alpha > 1$}}

We wish to obtain an inequality similar to \eqref{eq:d2} for 
$L_{\alpha}(\bQ(\cdot)) = \|\bQ(\cdot)\|_{1+\alpha}$ under the MW-$\alpha$ policy, and we accomplish this using the second-order mean value theorem (cf.\ Proposition \ref{prp:mvt}). 
Throughout this proof, we will drop the subscript $\alpha+1$ 
and use the notation $\|\cdot\|$ instead of $\|\cdot\|_{\alpha+1}$.

Consider the norm function 
$$g(\mathbf{x})=\| \mathbf{x}\|=(x_1^{\alpha+1}+\ldots+x_M^{\alpha+1})^{\frac{1}{\alpha+1}}.$$ 
The first derivative is 
$$\nabla g(\mathbf{x}) = \|\bx\|^{-\alpha} (x_1^{\alpha},\ldots,x_M^{\alpha}) 
~=~\frac{\bx^\alpha}{\|\bx\|^\alpha}.$$ 
Let $H(\bx) = [H_{ij}(\bx)]_{i,j=1}^{M}$ be  the second derivative (Hessian) matrix of $g$.
Then, 
$$ H_{ij}(\bx) ~=~ \frac{\partial^2 g }{\partial x_i \partial x_j} (\bx)
~=~\delta_{ij}\frac{\alpha x_i^{\alpha-1}}{\|\bx\|^{\alpha}}-\frac{\alpha x_i^{\alpha}x_j^{\alpha}}{\|\bx\|^{2\alpha+1}},$$
where $\delta_{ij}$ is the Kronecker delta. By Proposition \ref{prp:mvt}, 
for any $\bx, \by \in \Rp^M$, and with $\bdel = \by - \bx$, there exists 
a $\theta \in [0,1]$ for which
\begin{align*}
g(\by) & = g(\bx) + \bdel^T \nabla g(\bx) + \frac{1}{2}\bdel^T H(\bx + \theta \bdel) \bdel \\
& = g(\bx) + \|\bx\|^{-\alpha} \left(\sum_{i} \delta_i x_i^\alpha \right) \\ 
& ~~+ \frac{\alpha}{2} \|\bx + \theta \bdel\|^{-\alpha} \left(\sum_{i} (x_i + \theta \delta_i)^{\alpha-1} \delta^2_i \right) \\
& ~~- \frac{\alpha}{2} \|\bx + \theta \bdel\|^{-1-2\alpha} \left( \sum_{i, j} 
(x_i + \theta \delta_i)^{\alpha} (x_j + \theta \delta_j)^{\alpha} \delta_i \delta_j \right) \\
& = g(\bx) + \|\bx\|^{-\alpha} \left(\sum_{i} \delta_i x_i^\alpha \right) \\
& ~~+ \frac{\alpha}{2} \|\bx + \theta \bdel\|^{-\alpha} \left(\sum_{i} (x_i + \theta \delta_i)^{\alpha-1} \delta^2_i \right) \\
& ~~- \frac{\alpha}{2}\|\bx + \theta \bdel\|^{-1-2\alpha} \left( \sum_{i} 
(x_i + \theta \delta_i)^{\alpha} \delta_i \right)^2.
\end{align*}
Using $\bx = \bQ(\tau)$, $\by = \bQ(\tau+1)$ and $\bdel(\tau) = \bQ(\tau+1) - \bQ(\tau)$, 
we have 
\begin{align}
\|\bQ(\tau+1)\| &= \|\bQ(\tau)\| + \left[\frac{\sum_i \delta_i(\tau) Q_i^\alpha(\tau)}{\|\bQ(\tau)\|^{\alpha}}\right]\nonumber \\
& ~~+ \frac{\alpha}{2}\left[\frac{\sum_i (Q_i(\tau) + \theta \delta_i(\tau))^{\alpha-1} \delta^2_i(\tau)}{\|\bQ(\tau) + \theta \bdel(\tau)\|^{\alpha}}\right] \nonumber \\
& ~~- \frac{\alpha}{2} \left[\frac{\left(\sum_{i} 
(Q_i(\tau) + \theta \delta_i(\tau))^{\alpha} \delta_i(\tau) \right)^2}{\|\bQ(\tau) + \theta \bdel(\tau)\|^{1+2\alpha}}
\right].
\label{eq:mvt1}
\end{align}
Therefore, using the fact that $\delta_i(\tau) \in \{-1,0,1\}$, we have
\begin{align}
& \|\bQ(\tau+1)\| -  \|\bQ(\tau)\| \nonumber \\
&\leq \left[\frac{\sum_i \delta_i(\tau) Q_i^\alpha(\tau)}{\|\bQ(\tau)\|^{\alpha}}\right]+ \frac{\alpha}{2}\left[\frac{\sum_i (Q_i(\tau) + \theta \delta_i(\tau))^{\alpha-1}}{\|\bQ(\tau) + \theta \bdel(\tau)\|^{\alpha}}\right]. \label{eq:d4}
\end{align}
We take  conditional expectations  of both sides, given $\mathbf{Q}(\tau)$. To bound the first term on the RHS, we use the definition of the 
MW-$\alpha$ policy, the bound \eqref{eq:lamb} on
$\blambda$,  and the argument used to establish
\eqref{eq:lamb-drift} in the  proof of Theorem \ref{thm:sw-drift} for 
$\alpha=1$ (with $w(\bQ(\tau))$ replaced by $w_\alpha(\bQ(\tau))$). We
obtain 
\begin{equation}
\E\left[\frac{\sum_i \delta_i(\tau) Q_i^\alpha(\tau)}{\|\bQ(\tau)\|^{\alpha}}~\Big|~\bQ(\tau)\right] \leq 
-(1-\rho)\frac{w_{\alpha}(\mathbf{Q}(\tau))}{\|\mathbf{Q}(\tau)\|^{\alpha}}.
\end{equation}
Note that 
\begin{eqnarray}\label{eq:d6}
\|\mathbf{Q}(\tau)\|^{\alpha} & \leq & \left(M Q_{\max}(\tau)^{\alpha+1}\right)^{\frac{\alpha}{\alpha+1}} \nonumber \\
& = & M^{\frac{\alpha}{\alpha+1}}Q_{\max}^{\alpha}(\tau),
\end{eqnarray}
and 
\begin{displaymath}
w_{\alpha}(\mathbf{Q}(\tau))\geq Q_{\max}^{\alpha}(\tau).
\end{displaymath}
Therefore, 
\begin{eqnarray}
\E\left[\frac{\sum_i \delta_i(\tau) Q_i^\alpha(\tau)}{\|\bQ(\tau)\|^{\alpha}} ~\Big|~\bQ(\tau)\right] & \leq &
-(1-\rho) M^{-\frac{\alpha}{1+\alpha}}.
\label{eq:drift_al}
\end{eqnarray}

Consider now the second term of the conditional expectation of the RHS of Inequality (\ref{eq:d4}). Since $\alpha > 1$, and $\delta_i(\tau) \in \{-1,0,1\}$,
the numerator of the expression inside the bracket satisfies
\begin{displaymath}
\sum_i (Q_i(\tau) + \theta \delta_i(\tau))^{\alpha-1} \leq M \left(Q_{\max}(\tau)+1\right)^{\alpha-1},
\end{displaymath}
and the denominator satisfies
\begin{displaymath}
\|\bQ(\tau)+\theta \bdel(\tau)\|^{\alpha} \geq \left([Q_{\max}(\tau)-1]^+\right)^{\alpha},
\end{displaymath}
where we use the notation $[c]^+=0 \vee c$.
Thus,
\begin{displaymath}
\frac{\alpha}{2}\left[\frac{\sum_i (Q_i(\tau) + \theta \delta_i(\tau))^{\alpha-1}}{\|\bQ(\tau) + \theta \bdel(\tau)\|^{\alpha}}\right] \leq \frac{\alpha}{2}\cdot\frac{M(Q_{\max}+1)^{\alpha-1}}{([Q_{\max}(\tau)-1]^+)^{\alpha}}.
\end{displaymath}
Now if $\|\bQ(\tau)\|$ is large enough, $Q_{\max}(\tau)$ is large enough, and
$\frac{\alpha}{2}\cdot\frac{M(Q_{\max}+1)^{\alpha-1}}{([Q_{\max}(\tau)-1]^+)^{\alpha}}$ can be made arbitrarily small. Thus, the conditional expectation of the second term on the RHS of (\ref{eq:d4}) can be made arbitrarily small for large enough $\|\bQ(\tau)\|$. This fact, together with Inequality \eqref{eq:drift_al}, implies that there exists
$B>0$ such that if $\|\bQ(\tau)\|>B$, then
\begin{displaymath}
\E\left[\|\bQ(\tau+1)\|  -  \|\bQ(\tau)\| ~\Big|~\bQ(\tau)\right] \leq
-\frac{1-\rho}{2} M^{-\frac{\alpha}{1+\alpha}}.
\end{displaymath}

\subsection{Proof of Theorem \ref{thm:sw-drift}: {\large $\alpha \in (0,1)$}}

The proof in this section is similar to that for the case $\alpha >1$.
We invoke Proposition \ref{prp:mvt} to write the drift term as a sum 
of terms, which we bound separately. Note that to use Proposition \ref{prp:mvt},
we need $L_{\alpha}$ to be twice continuously differentiable. Indeed,  by Lemma \ref{lm:prpt} (i),
$f_{\alpha}$ is continuously differentiable, so its antiderivative $F_{\alpha}$ is twice continuously
differentiable, and so is $L_{\alpha}$. Thus, by the second order mean value theorem,
we obtain an equation similar to Equation (\ref{eq:mvt1}):
\begin{align}
& L_{\alpha}(\bQ(\tau+1))-L_{\alpha}(\bQ(\tau)) \nonumber \\
&= \left[\frac{\sum_i\delta_i(\tau)f_{\alpha}(Q_i(\tau))}{L_{\alpha}^{\alpha}(\bQ(\tau))}\right]+\frac{1}{2}\left[\frac{\sum_i f'_{\alpha}(Q_i(\tau)+\theta\delta_i(\tau))\delta_i^2(\tau)}{L_{\alpha}^{\alpha}(\bQ(\tau)+\theta\bdel(\tau))}\right] \nonumber \\
&~~-\frac{\alpha}{2}\left[\frac{(\sum_i\delta_i(\tau)f_{\alpha}(Q_i(\tau)+\theta\delta_i(\tau)))^2}{L_{\alpha}^{2\alpha+1}(\bQ(\tau)+\theta\bdel(\tau))}\right].
\label{eq:mvt2}
\end{align}

\noindent Again, using the fact $\delta_i(\tau)\in \{-1,0,1\}$,
\begin{displaymath}
L_{\alpha}(\bQ(\tau+1))-L_{\alpha}(\bQ(\tau))\leq T_1+T_2,
\end{displaymath}
where
\begin{displaymath}
T_1 = \frac{\sum_i\delta_i(\tau)f_{\alpha}(Q_i(\tau))}{L_{\alpha}^{\alpha}(\bQ(\tau))},
\end{displaymath}
and
\begin{displaymath}
T_2 = \frac{1}{2}\left[\frac{\sum_i f'_{\alpha}(Q_i(\tau)+\theta\delta_i(\tau))}{L_{\alpha}^{\alpha}(\bQ(\tau)+\theta\bdel(\tau))}\right].
\end{displaymath}

Let us consider $T_2$ first. For $\alpha \in (0,1)$, 
by Lemma \ref{lm:prpt} (iv), $f_{\alpha}'(r)\leq 2$ for all $r\geq 0$.
Thus
\begin{displaymath}
T_2 \leq \frac{1}{2}\left[\frac{2M}{L_{\alpha}^{\alpha}(\bQ(\tau)+\theta\bdel(\tau))}\right]=\frac{M}{L_{\alpha}^{\alpha}(\bQ(\tau)+\theta\bdel(\tau))}.
\end{displaymath}
which becomes arbitrarily small when  $L_{\alpha}(\bQ(\tau))$
is large enough.

We now consider
$T_1$. Since $f_{\alpha}(r)\leq r^{\alpha}+1$
for all $r\geq 0$ (cf. Lemma \ref{lm:prpt} (iii)), and $\delta_i(\tau)\in \{-1,0,1\}$,
\begin{displaymath}
T_1 \leq \frac{\sum_i\delta_i(\tau)Q_i^{\alpha}(\tau)}{L_{\alpha}^{\alpha}(\bQ(\tau))} +
         \frac{M}{L_{\alpha}^{\alpha}(\bQ(\tau))}.
\end{displaymath}
When we take the conditional expectation, an argument similar to the one for the case $\alpha > 1$ yields
\begin{equation}
\E\left[\frac{\sum_i \delta_i(\tau) Q_i^\alpha(\tau)}{L_{\alpha}^{\alpha}(\bQ(\tau))}~\Big|~\bQ(\tau)\right] \leq -(1-\rho)\frac{w_{\alpha}(\mathbf{Q}(\tau))}{L_{\alpha}^{\alpha}(\mathbf{Q}(\tau))}.
\label{eq:t1}
\end{equation}
Again, as before, $w_{\alpha}(\mathbf{Q}(\tau))\geq Q_{\max}^{\alpha}(\tau)$.
For the denominator, by Lemma \ref{lm:prpt} (iii'), for any $r\geq 0$, we have $(\alpha+1)F_{\alpha}(r)\leq r^{\alpha+1}+2$.
Thus 
\begin{eqnarray*}
L_{\alpha}(\bQ(\tau))& \leq & \left[\sum_i(Q_i(\tau)+2)^{\alpha+1}\right]^{\frac{1}{\alpha+1}} \\
&\leq & \left(M(Q_{\max}(\tau)+2)^{\alpha+1}\right)^{\frac{1}{\alpha+1}} \\
&=& M^{\frac{1}{\alpha+1}}(Q_{\max}(\tau)+2).
\end{eqnarray*}
Therefore, 
\begin{equation*}
\E\left[\frac{\sum_i \delta_i(\tau) Q_i^\alpha(\tau)}{L_{\alpha}^{\alpha}\bQ(\tau)}~\Big|~\bQ(\tau)\right]\leq
-(1-\rho)M^{-\frac{\alpha}{\alpha+1}}\frac{Q_{\max}^{\alpha}(\tau)}{(Q_{\max}+2)^{\alpha}}.
\end{equation*}
If $Q_{\max}(\tau)$ is large enough, we can further upper bound the RHS by, say, $-\frac{3}{4}(1-\rho)M^{-\frac{\alpha}{\alpha+1}}$.

Putting everything together, we have
\begin{align}
& \E\left[L_{\alpha}(\bQ(\tau+1))-L_{\alpha}(\bQ(\tau)) ~\Big |~ \bQ(\tau)\right] \nonumber \\
&\leq -\frac{3}{4}(1-\rho)M^{-\frac{\alpha}{1+\alpha}}+\frac{M}{L_{\alpha}^{\alpha}(\bQ(\tau))}+\E[T_2\mid \bQ(\tau)],
\end{align}
if $Q_{\max}(\tau)$ is large enough. As before, if $L_{\alpha}(\bQ(\tau))$ is large enough,
then $Q_{\max}(\tau)$ is large enough, and
$T_2$ and $\frac{M}{L_{\alpha}^{\alpha}(\bQ(\tau))}$ can be made arbitrarily small. 
Thus, there exists $B>0$ such that if $L_{\alpha}(\bQ(\tau))>B$, then
\begin{displaymath}
\E\left[L_{\alpha}(\bQ(\tau+1))-L_{\alpha}(\bQ(\tau)) ~\Big |~ \bQ(\tau)\right] \leq -\frac{1}{2}(1-\rho)M^{-\frac{\alpha}{1+\alpha}}.
\end{displaymath}

\section{Proof of Theorem \ref{THM:SW1}}

In this section, we fix some $\alpha\in(0,1)$ and prove that
the MW-$\alpha$ policy induces finite 
steady-state expected queue lengths.
The key to our proof is the use of the Lyapunov function 
$\Phi(\bx)=L_{\alpha}^2(\bx)$. This is to be contrasted with  
the use of the standard Lyapunov function, $\sum_i x_i^{1+\alpha}$, in the
literature, or the ``norm''-Lyapunov function $L_{\alpha}(\bx)$ that
we used in establishing the drift inequality of Theorem \ref{thm:sw-drift}. 

Throughout the proof, we drop the subscript $\alpha$ from $L_{\alpha}$, $F_{\alpha}$, and $f_{\alpha}$, 
as they are clear from the context. We also use $\|\bx\|$ to denote the $(\alpha+1)$-norm
of the vector $\bx$, again dropping the subscript.

As usual, we consider the conditional expected drift at time $\tau$, 
$$ D(\bQ(\tau)) ~=~ \E\left[\Phi(\bQ(\tau+1))-\Phi(\bQ(\tau)) ~\Big|~ \bQ(\tau)\right].$$
Recall the notation $Q_{\max}(\tau)=\max\{Q_1(\tau),\ldots,Q_M(\tau)\}$.
Since for $Q_{\max}<2$, $D(\bQ(\tau))$ is bounded by a constant, we assume throughout
the proof that $Q_{\max}(\tau)\geq 2$. 
As in the proof of Theorem \ref{thm:sw-drift} for the case $\alpha \in (0,1)$,
we shall use the second order mean value theorem to obtain a bound on $D(\bQ(\tau))$. 
Using the definition $\Phi(\bx) = L^2(\bx)$, we have 
\begin{eqnarray}\label{eq:phi-partial1}
\left[\nabla \Phi(\bx)\right]_i & = & 2 L(\bx) \frac{\partial L(\bx)}{\partial x_i} ~=~ 2 f(x_i) L^{1-\alpha}(\bx), 
\end{eqnarray}
and 
\begin{align}\label{eq:phi-partial2}
\frac{\partial^2 \Phi}{\partial x_i \partial x_j} (\bx)
& = 2 \frac{\partial L(\bx)}{\partial x_i}  \cdot\frac{\partial L(\bx)}{\partial x_j} 
+ 2 L(\bx) \frac{\partial^2 L(\bx)}{\partial x_i \partial x_j} \nonumber \\
& = 2 \frac{f(x_i) f(x_j)}{L^{2\alpha}(\bx)} + 2 L(\bx) \left(\delta_{ij}\frac{f'(x_i)}{L^{\alpha}(\bx)}-\frac{\alpha f(x_i)f(x_j)}{L^{2\alpha+1}(\bx)} \right) \nonumber \\
& = 2 (1-\alpha) \frac{f(x_i) f(x_j)}{L^{2\alpha}(\bx)} + 2 \delta_{ij} f'(x_i) L^{1-\alpha}(\bx).
\end{align}
Using the second order mean value theorem and the notation $\bQ(\tau+1) = \bQ(\tau) + \bdel(\tau)$, 
we have, for some $\theta \in [0,1]$,
\begin{align}
& \Phi(\bQ(\tau+1)) - \Phi(\bQ(\tau)) \nonumber \\
& \leq 2L^{1-\alpha}(\bQ(\tau)) \left(\sum_{i} f(Q_i(\tau))\delta_i(\tau) \right) \nonumber\\
& ~~+ L^{1-\alpha}(\bQ(\tau) +\theta\bdel(\tau)) \left(\sum_{i} f'(Q_i(\tau)+\theta\delta_i(\tau)) \right) \nonumber \\
& ~~+ (1-\alpha) \frac{\left(\sum_{i} f(Q_i(\tau)+\theta \delta_i(\tau)) \delta_i(\tau)\right)^2}{L^{2\alpha}(\bQ(\tau)+\theta\bdel(\tau))}. \label{eq:ph1}
\end{align}
Let us denote the three terms on the RHS of \eqref{eq:ph1} as $\bar{T}_1$, $\bar{T}_2$ and $\bar{T}_3$
respectively, so that
\begin{align*}
\bar{T}_1 & = 2L^{1-\alpha}(\bQ(\tau)) \left(\sum_{i} f(Q_i(\tau))\delta_i(\tau) \right),\\
\bar{T}_2 & = L^{1-\alpha}(\bQ(\tau) +\theta\bdel(\tau)) \left(\sum_{i} f'(Q_i(\tau)+\theta\delta_i(\tau)) \right),\\
\textrm{and \ } \bar{T}_3 & = (1-\alpha) \frac{\left(\sum_{i} f(Q_i(\tau)+\theta \delta_i(\tau)) \delta_i(\tau)\right)^2}{L^{2\alpha}(\bQ(\tau)+\theta\bdel(\tau))}.
\end{align*}
We consider these terms one at a time.  
\newline
\newline
\textbf{a)} By Lemma \ref{lm:prpt} (iii), $f(r)\leq r^{\alpha}+1$. Using the fact that $\delta_i(\tau)\in \{-1,0,1\}$, we obtain
\begin{displaymath}
\bar{T}_1\leq 2L^{1-\alpha}(\bQ(\tau))\left(M+\sum_i Q_i^{\alpha}(\tau)\delta_i(\tau)\right).
\end{displaymath}
When we take a conditional expectation, an argument similar to the one in earlier sections yields
\begin{displaymath}
\E\left[\sum_i Q_i^{\alpha}(\tau)\delta_i(\tau) ~\Big |~ \bQ(\tau) \right]\leq -(1-\rho)w_{\alpha}(\bQ(\tau)).
\end{displaymath}
Thus, 
\begin{eqnarray*}
\E\left[\bar{T}_1 ~\Big |~ \bQ(\tau) \right]&\leq & -2(1-\rho)w_{\alpha}(\bQ(\tau))L^{1-\alpha}(\bQ(\tau))\\
                                            & & +2M L^{1-\alpha}(\bQ(\tau)).
\end{eqnarray*}
In general, for $r, s\geq 0$ and $\beta\in [0,1]$,
\begin{eqnarray}\label{eq:normeq}
(r+s)^{\beta} & \leq & r^{\beta}+s^{\beta}.
\end{eqnarray}
Now, by Lemma \ref{lm:prpt} (iii'), $r^{\alpha+1}-2\leq (\alpha+1)F(r)\leq r^{\alpha+1}+2$, so
\begin{displaymath}
\sum_i x_i^{\alpha+1}-2M \leq (\alpha+1)\sum_i F(x_i)\leq \sum_i x_i^{\alpha+1}+2M.
\end{displaymath}
We use inequality \eqref{eq:normeq}, with $r = x_i^{\alpha+1}$, $s = 2M$, and   $\beta=(1-\alpha)/(1+\alpha) \in(0,1)$, to obtain
\begin{eqnarray*}
L^{1-\alpha}(\bx) 
&=& \left((\alpha+1)\sum_i F(x_i)\right)^{\frac{1-\alpha}{1+\alpha}} \\
&\leq & \left(2M+\sum_i x_i^{\alpha+1}\right)^{\frac{1-\alpha}{1+\alpha}} \\
&\leq & (2M)^{\frac{1-\alpha}{1+\alpha}}+\left(\sum_i x_i^{\alpha+1}\right)^{\frac{1-\alpha}{1+\alpha}}\\
& =& (2M)^{\frac{1-\alpha}{1+\alpha}}+\|\bx\|^{1-\alpha}.
\end{eqnarray*}
A similar argument, based on inequality \eqref{eq:normeq}, with $r = (\alpha+1)F(x_i)$ and $s = 2M$, yields
\begin{eqnarray*}
\|\bx\|^{1-\alpha} - (2M)^{\frac{1-\alpha}{1+\alpha}} & \leq & L^{1-\alpha}(\bx).
\end{eqnarray*}
We also know that
\begin{eqnarray*}
w_{\alpha}(\bQ(\tau)) ~~\geq~~ Q_{\max}^{\alpha}(\tau) & \geq & M^{-\frac{\alpha}{\alpha+1}}\|\bQ(\tau)\|^{\alpha}.
\end{eqnarray*}
Putting all these facts together, we obtain
\begin{align}
& \E\left[\bar{T}_1 ~\Big |~ \bQ(\tau) \right] \nonumber \\
&\leq 
- 2(1-\rho)w_{\alpha}(\bQ(\tau))L^{1-\alpha}(\bQ(\tau))+2M L^{1-\alpha}(\bQ(\tau))\nonumber \\
&\leq - 2(1-\rho)M^{-\frac{\alpha}{\alpha+1}}\|\bQ(\tau)\|^{\alpha}
          \left(\|\bQ(\tau)\|^{1-\alpha}-(2M)^{\frac{1-\alpha}{1+\alpha}}\right) \nonumber \\
& ~~+ 2 M\left((2M)^{\frac{1-\alpha}{1+\alpha}}+\|\bQ(\tau)\|^{1-\alpha}\right) \nonumber \\
& = - 2(1-\rho)M^{-\frac{\alpha}{\alpha+1}}\|\bQ(\tau)\|+2M\|\bQ(\tau)\|^{1-\alpha} \nonumber \\
& ~~+ 2^{\frac{2}{1+\alpha}}(1-\rho)M^{\frac{1-2\alpha}{1+\alpha}}\|\bQ(\tau)\|^{\alpha}+(2M)^{\frac{2}{1+\alpha}}.
\label{eq:tbar1}
\end{align}
\newline
\newline
\textbf{b)} We now consider the term $\bar{T}_2$. Since $\alpha\in (0,1)$, we have $f'(r)\leq 2$ for all $r\geq 0$.
Since we also have $\theta\in [0,1]$ and $\delta_i(\tau)\in \{-1,0,1\}$, and using the fact
that $L^{1-\alpha}(\bx)\leq (2M)^{\frac{1-\alpha}{1+\alpha}}+\|\bx\|^{1-\alpha}$, we have 
\begin{eqnarray*}
\bar{T}_2 & \leq & 2M L(\bQ(\tau) +\theta\bdel(\tau))^{1-\alpha} \\
          & \leq & 2M \left((2M)^{\frac{1-\alpha}{1+\alpha}}+\|\bQ(\tau)+\theta\bdel(\tau)\|^{1-\alpha}\right) \\
          & = & (2M)^{\frac{2}{1+\alpha}}+2M \|\bQ(\tau)+\theta\bdel(\tau)\|^{1-\alpha}.
\end{eqnarray*}
Now $\|\bQ(\tau)+\theta\bdel(\tau)\|\leq \|\bQ(\tau)+\bone\|\leq \|\bQ(\tau)\|+\|\bone\| = \|\bQ(\tau)\|+M^{\frac{1}{\alpha+1}}$. Since $\alpha \in (0,1)$, we have $0<1-\alpha<1$,
and so
\begin{eqnarray*}
\|\bQ(\tau)+\theta\bdel(\tau)\|^{1-\alpha}&\leq & \left(\|\bQ(\tau)\|+M^{\frac{1}{\alpha+1}}\right)^{1-\alpha} \\
                                          &\leq & \|\bQ(\tau)\|^{1-\alpha}+M^{\frac{1-\alpha}{\alpha+1}}.
\end{eqnarray*}
Putting everything together, we have
\begin{eqnarray}
\bar{T}_2 & \leq & (2M)^{\frac{2}{1+\alpha}}+2M \Big(\|\bQ(\tau)\|^{1-\alpha}+M^{\frac{1-\alpha}{\alpha+1}}\Big)\nonumber \\
          & = & (2+2^{\frac{2}{1+\alpha}})\alpha M^{\frac{2}{1+\alpha}}+2M\|\bQ(\tau)\|^{1-\alpha}.
\label{eq:tbar2}
\end{eqnarray}
\newline
\textbf{c)} We finally consider $\bar{T}_3$. For notational convenience, we write
$\bx = \bQ(\tau)+\theta\bdel(\tau)$, and let $x_{\max} = \max \{x_1,\ldots,x_M\}$.
Note that since $\delta_i(\tau)\in\{-1,0,1\}$, $\theta\in [0,1]$, and we assumed that
$Q_{\max}\geq 2$, we always have $x_{\max}\geq 1$.
We consider the numerator and the denominator separately.
First use the facts that $f(r)\geq 0$ for all $r\geq 0$ (cf. Lemma \ref{lm:prpt} (ii)),
and $\delta_i(\tau)\in \{-1,0,1\}$, to obtain
\begin{displaymath}
\left(\sum_i f(x_i)\delta_i(\tau)\right)^2 \leq \left(\sum_i f(x_i)\right)^2.
\end{displaymath}
Since $f$ is increasing in $r$ (cf. Lemma \ref{lm:prpt} (ii)), 
\begin{displaymath}
\left(\sum_i f(x_i)\right)^2 \leq \left(Mf(x_{\max})\right)^2 = M^2f^2(x_{\max}).
\end{displaymath}
Thus, 
\begin{displaymath}
\left(\sum_i f(x_i)\delta_i(\tau)\right)^2 \leq M^2 f^2(x_{\max}).
\end{displaymath}
Next, since $F(r) = \int_0^r f(s)\, ds$ and $f\geq 0$, we have $F\geq 0$ as well.
Thus,
\begin{displaymath}
L^{2\alpha}(\bx) = \left((\alpha+1)\sum_i F(x_i)\right)^{\frac{2\alpha}{\alpha+1}}
              \geq \left( (\alpha+1)F(x_{\max}) \right)^{\frac{2\alpha}{\alpha+1}},
\end{displaymath}
and so
\begin{displaymath}
\bar{T}_3 \leq (1-\alpha)\frac{M^2f^2(x_{\max})}{\left( (\alpha+1)F(x_{\max}) \right)^{\frac{2\alpha}{\alpha+1}}}.
\end{displaymath}
We will show that $\bar{T}_3$ is bounded above by a positive constant,
whenever $x_{\max}\geq 1$.
Indeed, by Lemma \ref{lm:prpt} (iii) and (iii'), as $x_{\max} \rightarrow \infty$,
\begin{displaymath}
\frac{f^2(x_{\max})}{x^{2\alpha}_{\max}}\rightarrow 1 \textrm{ \ and \ }
\frac{\left( (\alpha+1)F(x_{\max}) \right)^{\frac{2\alpha}{\alpha+1}}}{x^{2\alpha}_{\max}}\rightarrow 1,
\end{displaymath}
so
\begin{displaymath}
(1-\alpha)\frac{M^2f^2(x_{\max})}{\left( (\alpha+1)F(x_{\max}) \right)^{\frac{2\alpha}{\alpha+1}}} \rightarrow (1-\alpha)M^2
\end{displaymath}
as $x_{\max} \rightarrow \infty$. 
\old{Therefore it remains to show that as $x_{\max} \rightarrow 0$, $\bar{T}_3$ remains bounded.
Since $f(0)=F(0)=0$, we resort to L'Hopit\^{a}l's rule:
\begin{eqnarray*}
\lim_{r\rightarrow 0} \frac{f^2(r)}{F^{\frac{2\alpha}{1+\alpha}}(r)}
& = & \lim_{r\rightarrow 0} \frac{\frac{d}{dr}\left[f^2(r)\right]}{\frac{d}{dr}\left[F^{\frac{2\alpha}{1+\alpha}}(r)\right]}\\
& = & \lim_{r\rightarrow 0} \frac{2f(r)f'(r)}{\frac{2\alpha}{\alpha+1}F^{\frac{2\alpha}{\alpha+1}-1}(r)f(r)} \\
& = & \lim_{r\rightarrow 0} \frac{\alpha+1}{\alpha}f'(r)F^{\frac{1-\alpha}{\alpha+1}}(r) = 0,
\end{eqnarray*}
since $f'(0)=1$, $F(0)=0$, $\alpha>0$, and $1-\alpha>0$.}
Using the continuity of $f$ and $F$ for $x_{\max}\geq 1$, it follows 
that there exists a constant $\tilde{K}>0$ such that
\begin{equation}
\bar{T}_3\leq (1-\alpha)\frac{M^2f^2(x_{\max})}{\left((\alpha+1)F(x_{\max})\right)^{\frac{2\alpha}{\alpha+1}}} \leq \tilde{K},
\label{eq:tbar3}
\end{equation}
whenever $x_{\max}\geq 1$. 

Putting together the bounds \eqref{eq:tbar1}, \eqref{eq:tbar2}, and \eqref{eq:tbar3} 
for $\bar{T}_1$, $\bar{T}_2$, and $\bar{T}_3$, respectively, we conclude that, for $x_{\max}\geq 1$,
\begin{align}
D(\bQ(\tau))
&\leq -2(1-\rho)M^{-\frac{\alpha}{\alpha+1}}\|\bQ(\tau)\|+2M\|\bQ(\tau)\|^{1-\alpha} \nonumber \\
& ~~+4(1-\rho)M^{\frac{1-2\alpha}{1+\alpha}}\|\bQ(\tau)\|^{\alpha}+(2M)^{\frac{2}{1+\alpha}} \nonumber \\
& ~~+(2+2^{\frac{2}{1+\alpha}})\alpha M^{\frac{2}{1+\alpha}}+2M\|\bQ(\tau)\|^{1-\alpha}+\tilde{K} \nonumber \\
&=-\bar{A}\|\bQ(\tau)\|+C_1\|\bQ(\tau)\|^{1-\alpha}+C_2\|\bQ(\tau)\|^{\alpha}+  K,
\end{align}
for some positive constants $\bar A$, $C_1$, $C_2$ and $K$.
Since $\alpha\in(0,1)$, the $\|\bQ(\tau)\|$ term dominates. In particular, there exist positive constants $A$ and $D$ such that 
as long as $\max_i Q_i(\tau)\geq D$, we have
\begin{equation}\label{eq:drr}
D(\bQ(\tau))\leq -A \|\bQ(\tau)\| +  K.
\end{equation}
On the other hand, on the bounded set where $\max_i Q_i(\tau)\leq D$,
the drift $D(\bQ(\tau))$ is also bounded by a constant. By suitably redefining the constant $K$, we conclude that Eq.\ \eqref{eq:drr} holds for all possible values of $\bQ(\tau)$.

The drift condition \eqref{eq:drr} is the standard Foster-Lyapunov criterion for the Lyapunov
function $\Phi$ and implies the positive recurrence of the Markov chain $\bQ(\cdot)$
under the MW-$\alpha$ policy, for $\alpha \in (0,1)$. The irreducibility and aperiodicity of the underlying
Markov chain implies the existence of a unique stationary distribution $\bpi$ as well as ergodicity. Let $\bQ_{\infty}$ be a random variable distributed according to $\bpi$. Then $\bQ(\tau)$ converges to $\bQ_{\infty}$ in distribution. 
Using Skorohod's representation theorem, we can embed the random vectors $\bQ(\tau)$ 
in a suitable probability space
so that they converge to $\bQ_{\infty}$ almost surely. 
With this embedding, 
$\|\bQ(\tau)\|\rightarrow \|\bQ_{\infty}\|$, and
$(\sum_{\tau=0}^{T-1}\|\bQ(\tau)\|)/T \rightarrow \|\bQ_{\infty}\|$,
almost surely.
Using Fatou's Lemma, we have
\begin{eqnarray*}
\E\left[\|\bQ_{\infty}\|\right]&=&\E\left[\liminf_{T\to\infty} \frac{1}{T} \sum_{\tau=0}^{T-1}\|\bQ(\tau)\|\right] \\
                           & \leq & \liminf_{T} \E\left[\frac{1}{T} \sum_{\tau=0}^{T-1}\|\bQ(\tau)\|\right]. 
\end{eqnarray*}
On the other hand, the drift inequality \eqref{eq:drr} is well known to imply that the RHS above is finite; see, e.g., Lemma 4.1 of \cite{GNT06}. This proves that $\E\left[\|\bQ_{\infty}\|\right]<\infty$.
By the equivalence of norms, the result for $\|\bQ\|_1$ follows as well.
\old{
Let
$$ \bar{Q}(T)= \E\left[\frac{1}{T} \sum_{\tau=0}^{T-1}\|\bQ(\tau)\|\right].$$
Then, using Jensen's inequality twice, we have
\begin{eqnarray}
0 &\leq & \E\left[\Phi(\bQ(T))-\Phi(\bQ(0))\right] \nonumber \\
& \leq &  -A T \bar{Q}(T) + C_1 \left(\sum_{\tau=0}^{T-1} 
\E\left[\|\bQ(\tau)\|^{1-\alpha}\right] \right)+ C_2 \left(\sum_{\tau=0}^{T-1} 
\E\left[\|\bQ(\tau)\|^{\alpha}\right] \right) + T K \nonumber \\
& \leq & - A T \bar{Q}(T)  + C_1 \left(\sum_{\tau=0}^{T-1} \E\left[\|\bQ(\tau)\|\right]^{1-\alpha}\right) + C_2 \left(\sum_{\tau=0}^{T-1} \E\left[\|\bQ(\tau)\|\right]^{\alpha}\right) + T K \nonumber \\
& \leq & - A T \bar{Q}(T)+ C_1 T \bar{Q}(T)^{1-\alpha} + C_2 T \bar{Q}(T)^{\alpha} + T K,
\end{eqnarray}
if we start at $\bQ(0)\equiv 0$. From above, it follows immediately that $\bar{Q}(T)$ is uniformly
upper bounded for all $T$. Therefore, using Fatou's Lemma we have
\begin{eqnarray*}
\E\left[\|\bQ_{\infty}\|\right]&=&\E\left[\liminf_{T\to\infty} \frac{1}{T} \sum_{\tau=0}^{T-1}\|\bQ(\tau)\|\right] \\
                           & \leq & \liminf_{T} \E\left[\frac{1}{T} \sum_{\tau=0}^{T-1}\|\bQ(\tau)\|\right] \\
                           & < & \infty. 
\end{eqnarray*}
}

\section{Exponential Bound under MW-{\Large $\alpha$}}\label{sec:exp}

In this section we derive an exponential upper bound on the tail
probability of the stationary queue-size distribution, under the MW-$\alpha$ policy. 

\subsection{Proof of Theorem \ref{thm:sw2}: {\large $\alpha \geq 1$}}

The proof of Theorem \ref{thm:sw2} relies on the following 
proposition, 
and the drift inequality established in
Theorem \ref{thm:sw-drift}. 

\begin{proposition}\label{prop:sw-exp}
Consider a switched network operating under the MW-$\alpha$ policy 
with $\alpha \geq 1$, and arrival rate vector $\blambda$ with
$\rho = \rho(\blambda) < 1$.  Let $\bpi$ be the
Let 
unique stationary distribution of the Markov chain $\bQ(\cdot)$. 
Suppose that for all $\tau$,
$$ \Big | \|\mathbf{Q}(\tau+1)\|_{\alpha+1} - \|\mathbf{Q}(\tau)\|_{\alpha+1} \Big | \leq \nu_{\max}.$$
Furthermore, suppose that
for some constants $B>0$ and $\gamma>0$, and whenever $\|\mathbf{Q}(\tau)\|_{1+\alpha}>B$, we have 
\begin{displaymath}
\mathbb{E}[\|\mathbf{Q}(\tau+1)\|_{\alpha+1} - \|\mathbf{Q}(\tau)\|_{\alpha+1} \Big | \mathbf{Q}(\tau)]\leq -\gamma.
\end{displaymath}
Then for any $\ell \in \Zp$, 
\begin{displaymath}
\mathbb{P}_{\bpi}\big( \|\mathbf{Q}(\tau)\|_{\alpha+1}> B + 2 \nu_{\max}\ell\big)\leq \Big(\frac{\bar{\nu}}{\bar{\nu}+\gamma}\Big)^{\ell+1}.
\end{displaymath}
\end{proposition}
Proposition \ref{prop:sw-exp} follows immediately  from the 
following Lemma, which is a minor adaptation of Lemma 1 of \cite{BGT01}. An 
interested reader may refer to the proof of Theorem 1(a) in \cite{BGT01} to see how
Lemma \ref{lem:sw-exp} leads to the bound claimed in Proposition \ref{prop:sw-exp}. 
\begin{lemma}\label{lem:sw-exp}
Under the same assumptions in Proposition \ref{prop:sw-exp}, and for any $c>B-\nu_{\max}$, 
\begin{align}
& \mathbb{P}_{\bpi}\left( \|\mathbf{Q}(\tau)\|_{\alpha+1}> c+\nu_{\max}\right)\nonumber \\
& \leq \left(\frac{\bar{\nu}}{\bar{\nu}+\gamma}\right)\mathbb{P}_{\bpi}\left(\|\mathbf{Q}(\tau)\|_{\alpha+1}> c-\nu_{\max}\right).
\end{align}
\end{lemma}
\begin{proof}
Since this Lemma is a minor adaptation of Lemma 1 in \cite{BGT01}, 
we only indicate the changes to the proof of Lemma
1 in \cite{BGT01} that lead to our claimed result. First let us point out
that the proof in \cite{BGT01} makes use of the finiteness of the expected value of 
the Lyapunov function under the stationary distribution $\bpi$. In our case, 
the Lyapunov function in question is $\|\cdot\|_{\alpha+1}$, and the finiteness follows 
from Theorem \ref{THM:SW1} by noticing that all norms are equivalent.

As in the proof of Lemma 1 in \cite{BGT01}, define 
$\hat{\Phi}(\mathbf{x})=\max \{c,\|\mathbf{x}\|_{\alpha+1}\}$. Note that
the maximal change in  $\Phi(\bx)$ in one time step 
is at most $\nu_{\max}$. 
As in \cite{BGT01}, we consider all $\mathbf{x}$ satisfying 
$c-\nu_{\max} < \|\mathbf{x}\|_{\alpha+1} \leq c+\nu_{\max}$. Then, 
\begin{align*}
& \mathbb{E}\big [\hat{\Phi}(\mathbf{Q}(\tau+1))|\mathbf{Q}(\tau) = \mathbf{x}\big ]-\hat{\Phi}(\mathbf{x}) \\
&\leq \sum_{\mathbf{x}':\|\mathbf{x}'\|>\|\mathbf{x}\|}p(\mathbf{x},\mathbf{x}')(\|\mathbf{x}'\|-\|\mathbf{x}\|)\\
&\leq \mathbb{E}\big [ \|\mathbf{a}(\tau)\|\big ] ~=~ \bar{\nu}.
\end{align*}
The proof of Lemma 1 in \cite{BGT01} esentially used $\nu_{\max}$ as an upper bound on $\bar{\nu}$. 
For our result, we keep $\bar{\nu}$ and then proceed as in the proof in \cite{BGT01}.
\end{proof}

\paragraph{Completing the Proof of Theorem \ref{thm:sw2} ($\alpha\geq 1$)} Now the proof of
Theorem \ref{thm:sw2} follows immediately from Proposition \ref{prop:sw-exp}
by noticing that 
Theorem \ref{thm:sw-drift} provides the desired
drift inequality, and the maximal change in $\|\bQ(\tau)\|_{1+\alpha}$ in one 
time step is at most $\nu_{\max} = M^{\frac{1}{1+\alpha}}$, because each queue
can receive at most one arrival and have at most one departure per time step.

\subsection{Proof of Theorem \ref{thm:sw2}: {\large $\alpha \in (0,1)$}}
The proof for the case $\alpha \in (0,1)$ is entirely
parallel to that in the previous section and we do not reproduce it here.

\section{Transient Analysis}
In this section, we present a transient analysis of the MW-$\alpha$ policy with $\alpha\geq 1$.
First we present a general maximal lemma, which is then specialized to the switched network.
In particular, we prove a drift inequality for the Lyapunov function $\tilde{L}(\bx)=\frac{1}{\alpha+1}\sum_i x_i^{\alpha+1}$.
We combine the drift inequality with the maximal lemma to obtain a maximal inequality for the switched network.
We then apply the maximal inequality to prove full state space collapse for $\alpha\geq 1$.

\subsection{The Key Lemma}
Our analysis relies on the following lemma:
\begin{lemma}\label{maximal}
Let $(\mathscr{F}_n)_{n\in \Zp}$ be a filtration on a probability space. Let $(X_n)_{n\in\Zp}$ be a nonnegative $\mathscr{F}_n$-adapted stochastic process that satisfies
\begin{equation}
\mathbb{E}[ X_{n+1}\mid \mathscr{F}_n ]\leq X_n+B_n
\label{eq:maxdrift}
\end{equation}
where $B_n$'s are nonnegative random variables (not necessarily $\mathscr{F}_n$-adapted) with finite means.
Let $X_n^{*} = \max\{X_0,\ldots,X_n\}$ and suppose that $X_0 = 0$. Then, for any $a>0$ and any $T\in \Zp$,
\begin{displaymath}
\mathbb{P}(X_T^{*} \geq a)\leq \frac{\sum_{n=0}^{T-1}\mathbb{E}[B_n]}{a}.
\end{displaymath}
\end{lemma}
This lemma is a simple consequence of the following standard maximal inequality for nonnegative supermartingales (see for example, Exercise 4, Section 12.4, of \cite{grimmett}):
\begin{theorem}\label{thm:smart}
Let $(\mathscr{F}_n)_{n\in\Zp}$ be a filtration on a probability space. Let $(Y_n)_{n\in \Zp}$
 be a nonnegative $\mathscr{F}_n$-adapted supermartingale, i.e., for all $n$,
$$\E[Y_{n+1}\mid \mathscr{F}_n ]\leq Y_n.$$
Let $Y_T^* = \max\{Y_0,\ldots,Y_T\}$. Then,
$$\mathbb{P}(Y_T^*\geq a)\leq \frac{\E[Y_0]}{a}.$$
\end{theorem}
\begin{proof}[of Lemma \ref{maximal}] First note that if we take the conditional expectation
on both sides of \eqref{eq:maxdrift}, given $\mathscr{F}_n$, we have
\begin{eqnarray*}
\E[X_{n+1}\mid \mathscr{F}_n]& \leq & \E[X_n\mid \mathscr{F}_n]+\E[B_n\mid F_n]\\
& = & X_n + \E[B_n\mid F_n].
\end{eqnarray*}
Fix $T\in \Zp$. For any $n\leq T$, define
$$Y_n = X_n + \E\left[\sum_{k=n}^{T-1}B_k~\Big |~ \mathscr{F}_n\right].$$
Then
\begin{eqnarray*}
\E[Y_{n+1}\mid \mathscr{F}_n] & = & \E[X_{n+1}\mid \mathscr{F}_n] \\
& & + \E\left[ \E \left[ \sum_{k=n+1}^{T-1} B_k~\Big |~\mathscr{F}_{n+1}\right]~\Big |~ \mathscr{F}_n \right]\\
&\leq & X_n + \E[B_n\mid F_n]+ \E \left[ \sum_{k=n+1}^{T-1} B_k~\Big |~\mathscr{F}_n\right] \\
& = & Y_n.
\end{eqnarray*}
Thus, $Y_n$ is an $\mathscr{F}_n$-adapted supermartingale; furthermore, by definition, $Y_n$ is non-negative for all $n$.
 Therefore, by Theorem \ref{thm:smart},
\begin{displaymath}
\mathbb{P}(Y_T^*\geq a )\leq \frac{\E[Y_0]}{a} = \frac{\E\left[\sum_{k=0}^{T-1} B_k\right]}{a}.
\end{displaymath}
But $Y_n\geq X_n$ for all $n$, since the $B_k$ are nonnegative. Thus
$$\mathbb{P}(X_T^*\geq a)\leq \mathbb{P}(Y_T^*\geq a)\leq \frac{\E\left[\sum_{k=0}^{T-1} B_k\right]}{a}.$$
\end{proof}
\noindent We have the following corollary of Lemma \ref{maximal} in which we take all the $B_n$ equal to the same constant:
\begin{corollary}\label{dmax}
Let $\mathscr{F}_n$, $X_n$ and $X_n^{*}$ be as in Lemma \ref{maximal}. Suppose that
\begin{displaymath}
\mathbb{E}[X_{n+1}\mid\mathscr{F}_n]\leq X_n+B,
\end{displaymath}
for all $n\geq 0$, where $B$ is a nonnegative constant.
Then, for any $a>0$ and any $T\in \Zp$,
\begin{displaymath}
\mathbb{P}(X_T^{*} \geq a)\leq \frac{BT}{a}.
\end{displaymath}
\end{corollary}

\subsection{The Maximal Inequality for Switched Networks}
We employ the Lyapunov function
\begin{equation}
\tilde{L}(\mathbf{x})=\frac{1}{\alpha+1}\sum_{i=1}^M x_i^{\alpha+1},
\end{equation}
to study the MW-$\alpha$ policy. This is the Lyapunov function that was used in \cite{MAW96} to establish positive recurrence of the chain $\bQ(\cdot)$ under the MW-$\alpha$ policy. Below we fine-tune the proof in \cite{MAW96} to obtain a more precise bound.
\begin{lemma}\label{lem:driftbd}
Let $\alpha \geq 1$. For a switched network model operating under the MW-$\alpha$ policy with $\rho = \rho(\blambda)<1$,
we have:
\begin{equation}
\mathbb{E}\big[\tilde{L}(\mathbf{Q}(\tau+1))-\tilde{L}(\mathbf{Q}(\tau))~\big |~ \mathbf{Q}(\tau)\big]\leq \frac{\bar{K}(\alpha,M)}{(1-\rho)^{\alpha-1}},
\end{equation}
where $\bar{K}(\alpha,M)$ is a constant depending only on $\alpha$ and $M$.
\end{lemma}
\begin{proof}
We employ the same strategy as in previous sections. By the second-order mean value theorem, there exists $\theta \in [0,1]$ such that
\begin{align*}
& \tilde{L}(\mathbf{Q}(\tau+1))-\tilde{L}(\mathbf{Q}(\tau)) \\
&=\frac{1}{\alpha+1}\sum_{i=1}^M ((Q_i(\tau)+\delta_i(\tau))^{\alpha+1}-Q_i^{\alpha+1}(\tau))\\
&=\sum_{i=1}^M Q_i^{\alpha}(\tau)\delta_i(\tau)+\sum_{i=1}^M \alpha (Q_i(\tau)+\theta \delta_i(\tau))^{\alpha-1}\delta_i^2(\tau).
\end{align*}
Let us bound the second term on the RHS. We have
\begin{align*}
& \sum_{i=1}^M \alpha (Q_i(\tau)+\theta \delta_i(\tau))^{\alpha-1}\delta_i^2(\tau) \\
&\leq \sum_{i=1}^M \alpha (Q_i(\tau)+\theta)^{\alpha-1} \leq \sum_{i=1}^M \alpha (Q_i(\tau)+1)^{\alpha-1}\\
&\leq \alpha \sum_{i=1}^M (2^{\alpha-1}Q_i^{\alpha-1}(\tau)+1)=\alpha 2^{\alpha-1}\sum_{i=1}^M Q_i^{\alpha-1}(\tau)+\alpha M \\
&\leq \alpha 2^{\alpha-1}M Q_{\max}^{\alpha-1}(\tau)+\alpha M.
\end{align*}
The third inequality follows because when $Q_i(\tau)\geq 1$, $(Q_i(\tau)+1)^{\alpha-1}\leq (2 Q_i(\tau))^{\alpha-1} = 2^{\alpha-1}Q_i^{\alpha-1}(\tau)$, and when $Q_i(\tau)=0$, $(Q_i(\tau)+1)^{\alpha-1}=1$.\\\\
Let us now take conditional expectations. From Section \ref{sec:drift}, we know that
\begin{eqnarray*}
\mathbb{E}\left[\sum_{i=1}^M Q_i^{\alpha}(\tau)\delta_i(\tau) ~\big |~ \mathbf{Q}(\tau) \right] &\leq & -(1-\rho)w_{\alpha}(\mathbf{Q}(\tau))\\
&\leq & -(1-\rho)Q_{\max}^{\alpha}(\tau).
\end{eqnarray*}
Thus, if we combine the inequalities above, we have
\begin{align}
& \mathbb{E}\big[\tilde{L}(\mathbf{Q}(\tau+1))-\tilde{L}(\mathbf{Q}(\tau))\big | \mathbf{Q}(\tau)\big] \nonumber \\
&\leq -(1-\rho)Q_{\max}^{\alpha}(\tau)+\alpha 2^{\alpha-1}M Q_{\max}^{\alpha-1}(\tau)+\alpha M. \label{eq:drift3}
\end{align}
It is a simple exercise in calculus to see that the RHS of \eqref{eq:drift3} is maximized at $Q_{\max}(\tau)=(\alpha-1)2^{\alpha-1}M/(1-\rho)$, giving the maximum value $$\frac{(\alpha-1)^{\alpha-1}2^{\alpha(\alpha-1)}M^{\alpha}}{(1-\rho)^{\alpha-1}}+\alpha M=O((1-\rho)^{1-\alpha}).$$
\end{proof}

\paragraph{Proof of Theorem \ref{thm:max}}
Let $b>0$. Then
\begin{align*}
\mathbb{P}\left(Q^*_{\max}(T)\geq b\right) & = \mathbb{P}\left(\frac{1}{\alpha+1}\big(Q^*_{\max}(T)\big)^{\alpha+1}\geq \frac{1}{\alpha+1}b^{\alpha+1}\right)\\
&\leq \mathbb{P}\left(\max_{\tau\in \{0,\ldots,T\}}\tilde{L}(\bQ(\tau))\geq \frac{1}{\alpha+1}b^{\alpha+1}\right).
\end{align*}
Now, by Lemma \ref{lem:driftbd} and Corollary \ref{dmax},
\begin{align*}
\mathbb{P}\left(\max_{\tau\in \{0,\ldots,T\}}\tilde{L}(\bQ(\tau))\geq \frac{1}{\alpha+1}b^{\alpha+1}\right)& \leq \frac{(\alpha+1)\bar{K}(\alpha,M) T}{(1-\rho)^{\alpha-1}b^{\alpha+1}}\\
& = \frac{K(\alpha,M) T}{(1-\rho)^{\alpha-1}b^{\alpha+1}},
\end{align*}
where $K(\alpha,M) = (\alpha+1)\bar{K}(\alpha,M)$.

\subsection{Full State Space Collapse for {\large $\alpha \geq 1$}}\label{ssec:ssc}
Throughout this section, we assume that we are given $\alpha \geq 1$, and correspondingly, the Lyapunov function $\tilde{L}(\boldsymbol{x})=\frac{1}{\alpha+1}\sum_{i=1}^M x_i^{\alpha+1}$. To state the full state space collapse result for $\alpha \geq 1$, we need some preliminary definitions and the statement of the multiplicative state space collapse result.\\\\
Let $\Sigma$ be the convex hull of $\mathcal{S}$ (the set of feasible schedules), and let $\bar{\bLambda}$ be defined by
\begin{displaymath}
\bar{\bLambda} = \Big \{ \boldsymbol{\lambda}\in \mathbb{R}_{+}^M : \boldsymbol{\lambda}\leq \boldsymbol{\sigma'} \textrm{ componentwise, for some } \boldsymbol{\sigma'} \in \Sigma \Big \}.
\end{displaymath}
Note that this is the closure of the capacity region $\bLambda$ defined earlier.
Recall the definition of the \emph{load} $\rho(\boldsymbol{\lambda})$ of an arrival rate vector $\boldsymbol{\lambda}$. It is clear that $\boldsymbol{\lambda} \in \bar{\bLambda}$ iff $\rho(\boldsymbol{\lambda}) \leq 1$. Define $\partial \bLambda$ the set of \emph{critical} arrival rate vectors:
\begin{eqnarray*}
\partial \bLambda &=& \bar{\bLambda} - \bLambda = \Big\{ \boldsymbol{\lambda}\in \bar{\bLambda} : \rho(\boldsymbol{\lambda})=1 \Big\}.
\end{eqnarray*}
\noindent Now consider the linear optimization problem, named DUAL($\boldsymbol{\lambda}$) in \cite{SW09}:
\begin{center}
$\begin{array}{ll}
\textrm{maximize} & \boldsymbol{\xi}\cdot \boldsymbol{\lambda} \\
\textrm{subject to} & \max_{\boldsymbol{\sigma}\in \mathcal{S}}\boldsymbol{\xi}\cdot\boldsymbol{\sigma}\leq 1,\\  
                  & \boldsymbol{\xi} \in \mathbb{R}_+^{M}.
                  
\end{array}$
\end{center}
For $\boldsymbol{\lambda}\in \partial \bLambda$, the optimal value of the objective in DUAL($\boldsymbol{\lambda}$) is $1$ (cf.\cite{SW09}). The set of optimal solutions to DUAL($\boldsymbol{\lambda}$) is a bounded polyhedron, and we let $\mathcal{S}^{*}=\mathcal{S}^* (\boldsymbol{\lambda})$ be the set of its extreme points.

Fix $\boldsymbol{\lambda}\in \partial \bLambda$. We then consider the optimization problem ALGD($w$):
\begin{center}
$\begin{array}{ll}
\textrm{minimize} & \tilde{L}(\boldsymbol{x})\\
\textrm{subject to} & \boldsymbol{\xi}\cdot \boldsymbol{x}\geq w_{\boldsymbol{\xi}} \textrm{ for all } \boldsymbol{\xi}\in \mathcal{S}^*(\boldsymbol{\lambda}),\\
                  & \boldsymbol{x}\in \mathbb{R}_+^M.
\end{array}$
\end{center}
We know from \cite{SW09} that ALGD($w$) has a unique solution. We now define the \emph{lifting map}:
\begin{definition}
Fix some $\boldsymbol{\lambda}\in \partial \bLambda$. The \emph{lifting map} $\Delta^{\boldsymbol{\lambda}}:\mathbb{R}_+^{|\mathcal{S}^*(\boldsymbol{\lambda})|}\rightarrow \mathbb{R}_+^{M}$ maps $w$ to the unique solution to ALGD($w$). We also define the \emph{workload map} $W^{\boldsymbol{\lambda}}: \mathbb{R}_+^M\rightarrow \mathbb{R}_+^{|\mathcal{S}^*(\boldsymbol{\lambda})|}$ by $W^{\boldsymbol{\lambda}}(\mathbf{q})=(\boldsymbol{\xi}\cdot\mathbf{q})_{\boldsymbol{\xi}\in \mathcal{S}^*(\boldsymbol{\lambda})}$.
\end{definition}
\noindent Fix $\boldsymbol{\lambda}\in \partial \bLambda$. Consider a sequence of switched networks indexed by $r\in \mathbb{N}$, operating under the MW-$\alpha$ policy (recall that $\alpha\geq 1$ here), all with the same number $M$ of queues and feasible schedules. 
Suppose that $\boldsymbol{\lambda}^r \in \bLambda$ for all $r$, and that $\boldsymbol{\lambda}^r = \boldsymbol{\lambda}-\boldsymbol{\Gamma}/r$, for some $\boldsymbol{\Gamma}\in \mathbb{R}_+^M$. For simplicity, suppose that all networks start with empty queues. Consider the following central limit scaling, 
\begin{equation}
\hat{\mathbf{q}}^r(t)=\mathbf{Q}^r(r^2t)/r,
\label{eq:scaling}
\end{equation}
where $\mathbf{Q}^r(\tau)$ is the queue size vector of the $r$th network at time $\tau$, and where we extend the domain of $\mathbf{Q}^r(\cdot)$ to $\mathbb{R}_+$ by linear interpolation in each interval $(\tau-1,\tau)$.\\\\
We are finally ready to state the multiplicative state space collapse result (Theorem 8.2 in \cite{SW09}):
\begin{theorem}\label{thm:mssc}
Fix $T>0$, and let $$\|\mathbf{x}(\cdot)\| = \sup_{i\in \{1,\ldots,M\}, 0\leq t\leq T}|x_i(t)|.$$ Under the above assumptions, for any $\beps>0$,
\begin{displaymath}
\lim_{r\rightarrow \infty}\mathbb{P}\left ( \frac{\| \hat{\mathbf{q}}^r(\cdot)-\Delta^{\boldsymbol{\lambda}}(W^{\boldsymbol{\lambda}}(\hat{\mathbf{q}}^r(\cdot)))\|}{\|\hat{\mathbf{q}}^r(\cdot)\|\vee 1}<\beps \right ) = 1.
\end{displaymath}
\end{theorem}
\noindent We now state and prove the full state space collapse result.
\begin{theorem}\label{thm:ssc}
Under the same assumptions in Theorem \ref{thm:mssc}, and for any $\beps>0$,
\begin{displaymath}
\lim_{r\rightarrow \infty}\mathbb{P}\left ( \| \hat{\mathbf{q}}^r(\cdot)-\Delta^{\boldsymbol{\lambda}}(W^{\boldsymbol{\lambda}}(\hat{\mathbf{q}}^r(\cdot)))\|<\beps \right ) = 1.
\end{displaymath}
\end{theorem}
\begin{proof}
First note that since $\boldsymbol{\lambda}^r = \boldsymbol{\lambda}-\boldsymbol{\Gamma}/r$, the corresponding loads satisfy $\rho_r \leq 1-C/r$, for some positive constant $C>0$. 
By Theorem \ref{thm:max}, for any $b>0$,
\begin{eqnarray*}
\mathbb{P}\left( \max_{\tau\in \{0,1,\ldots,r^2T\}}Q^r_{\max}(\tau)\geq b \right)& \leq & \frac{K(\alpha,M)r^2T}{(1-\rho)^{\alpha-1}b^{\alpha+1}}\\
&\leq & \frac{K(\alpha,M) r^{1+\alpha} T}{C^{\alpha-1} b^{\alpha+1}}.
\end{eqnarray*}
Then with $a=b/r$ and under the scaling in \eqref{eq:scaling},
\begin{displaymath}
\mathbb{P}\big( \| \hat{\mathbf{q}}^r(\cdot)\|\geq a \big)\leq \frac{K(\alpha,M)}{C^{\alpha-1}}\frac{T}{a^{\alpha+1}},
\end{displaymath}
for any $a>0$.\\\\
For notational convenience, we write $$D(r) = \|\hat{\mathbf{q}}^r(\cdot)-\Delta^{\boldsymbol{\lambda}}(W^{\boldsymbol{\lambda}}(\hat{\mathbf{q}}^r(\cdot)))\|.$$ Then, for any $a>1$,
\begin{eqnarray*}
\mathbb{P}\big ( D(r)\geq \beps \big )&\leq & \mathbb{P}\left ( \frac{D(r)}{\|\hat{\mathbf{q}}^r(\cdot)\|}> \frac{\beps}{a} \textrm{ or } \| \hat{\mathbf{q}}^r(\cdot)\| \geq a \right )\\
&\leq & \mathbb{P}\left( \frac{D(r)}{\|\hat{\mathbf{q}}^r(\cdot)\|}> \frac{\beps}{a}\right )+\mathbb{P}\big (\|\hat{\mathbf{q}}^r(\cdot)\| \geq a \big ).
\end{eqnarray*}
Note that by Theorem \ref{thm:mssc}, the first term on the RHS goes to $0$ as $r\rightarrow \infty$, for any $a>0$. The second term on the RHS can be made arbitrarily small by taking $a$ sufficiently large. Thus, $\mathbb{P}(D(r)\geq \beps)\rightarrow 0$ as $r\rightarrow \infty$. This concludes the proof.
\end{proof}

\section{Discussion}\label{sec:conclu}
The results in this paper can be viewed from two different  perspectives. On the one
hand, they provide much new information on  the qualitative behavior (e.g., finiteness of
expected backlog, bounds  on steady-state tail probabilities and finite-horizon maximum 
excursion probabilities, etc.) of
MW-$\alpha$ policies for switched network models. On the other hand,  at a technical
level, our results highlight the importance of choosing  a suitable Lyapunov function:
even if a network is shown to be stable  by using a particular Lyapunov function,
different choices may lead to  more powerful bounds.

The methods and results in this paper extend in two directions. First,  all of the
results, suitably restated, remain valid for multihop  networks under
backpressure-$\alpha$ policies. Second, the same is  true for flow-level models of the
type considered in \cite{KW04}. These  extensions will be reported elsewhere. 
\bibliographystyle{abbrv}
\bibliography{bib1}

\newpage

\appendix
\section{Tightness of the Upper Bound}
Here we compare the exponential upper bound in Theorem 
\ref{thm:sw2} with a lower bound that we shall obtain from the LDP results in
\cite{VL09}. To be able to evaluate a useful lower bound explicitly,
we consider the case of an input-queued (IQ) switch. As discussed in 
Section \ref{ssec:switch}, in an $m$-port IQ switch, there are $M=m^2$
queues. Let $\blambda = (\lambda_{ij})$ be the vector of arrival rates. 
Then, the load $\rho = \rho(\blambda)$ is 
$$ \rho = \max \left(\max_j\left(\sum_{i=1}^m\lambda_{ij}\right), ~ \max_i\left(\sum_{j=1}^m\lambda_{ij}\right)\right).$$
\paragraph{Upper bound for IQ switch} First, we specialize 
Theorem \ref{thm:sw2} to the case of an IQ switch. 
\begin{proposition}\label{prop:iq}
Consider an $m$-port IQ switch operating under the MW-$\alpha$ policy,
with arrival vector $\blambda$ so that $\rho = \rho(\blambda) < 1$. 
Let $\bpi$ be the stationary distribution of the queue-size. Then, for
some large enough constants $B>0$ and $B'>0$, and for every $\ell \in \Zp$:
\begin{itemize}
\item[(i)] if $\alpha \geq 1$, then
\begin{displaymath}
\mathbb{P}_{\bpi}\left(\|\mathbf{Q}\|_{\alpha+1} > B+ 2 m^{\frac{2}{\alpha+1}}\ell\right)
\leq \left(\frac{1}{1+\frac{1-\rho}{2m}}\right)^{\ell+1},
\end{displaymath}
\item[(ii)] and if $\alpha \in (0,1)$, then
\begin{displaymath}
\mathbb{P}_{\bpi}\left(\|\mathbf{Q}\|_{\alpha+1} > B'+ 10 m^{\frac{2}{\alpha+1}}\ell\right)
\leq \left(\frac{1}{1+\frac{1-\rho}{10m}}\right)^{\ell+1}.
\end{displaymath}
\end{itemize}
\end{proposition}
Before providing the proof, we note that Proposition \ref{prop:iq} suggests 
that for all $\alpha > 0$, as $\rho \to 1$, and using $\log (1/(1+r)) \approx -r$ for small $r>0$, 
we have
\begin{equation}\label{eq:iq-ub}
\limsup_{R\rightarrow\infty}\frac{1}{R}\log \mathbb{P}_{\bpi}(||\mathbf{Q}(\tau)|| > R) \leq -\frac{1-\rho}{100}~m^{-1-\frac{2}{\alpha+1}}.
\end{equation}

\vspace{.05in}

\begin{proof}[of Proposition \ref{prop:iq}] 
We need to identify the maximal change $\nu_{\max}$,
the drift constant $\gamma$, and $\bar{\nu}$. Since $M=m^2$, 
$$ \|\bQ(\tau+1) - \bQ(\tau) \|_{1+\alpha} \leq m^{\frac{2}{1+\alpha}}$$
For $\alpha \geq 1$, $m^{\frac{2}{\alpha+1}}$ can serve as $\nu_{\max}$ for $\|\bQ(\cdot)\|_{\alpha+1}$.
For $\alpha \in (0,1)$, $5m^{\frac{2}{\alpha+1}}$ can serve as $\nu_{\max}$ for $L_{\alpha}(\bQ(\cdot))$.
The drift constant $\gamma = \frac{1-\rho}{2m^{\alpha/(1+\alpha)}}$ 
is obtained in Lemma \ref{lem:iq-exp} stated below. Finally, for $\bar{\nu}$, we have
$\bar{\nu} = \E\left[\|\mathbf{a}(1)\|_{1+\alpha}\right] \leq (\rho m)^{\frac{1}{1+\alpha}}\leq m^{\frac{1}{1+\alpha}}$. 
Thus, 
$$ \frac{\gamma}{\bar{\nu}} \geq \frac{1-\rho}{2 m}.$$
Putting everything together and applying Theorem \ref{thm:sw2}, 
we obtain for $\alpha \geq 1$ and some $B>0$,
\begin{eqnarray*}
\mathbb{P}_{\bpi}\left(\|\mathbf{Q}\|_{\alpha+1} > B +  2 m^{\frac{2}{\alpha+1}}\ell\right)
& \leq & \left(\frac{1}{1+\frac{1-\rho}{2m}}\right)^{\ell+1},
\end{eqnarray*}
and for $\alpha\in (0,1)$ and some $B'>0$,
\begin{eqnarray*}
\mathbb{P}_{\bpi}\left(\|\mathbf{Q}\|_{\alpha+1} > B' +  10 m^{\frac{2}{\alpha+1}}\ell\right)
& \leq & \left(\frac{1}{1+\frac{1-\rho}{10m}}\right)^{\ell+1},
\end{eqnarray*}
This completes the proof of Proposition \ref{prop:iq}.
\end{proof}
\begin{lemma}\label{lem:iq-exp}
Under the same assumptions of Proposition \ref{prop:iq}, and for all $\alpha > 0$,
there exists a constant $B>0$ such that
\begin{equation}
\mathbb{E}\big[L_{\alpha}(\mathbf{Q}(\tau+1))-L_{\alpha}(\mathbf{Q}(\tau)) ~\big|~\mathbf{Q}(\tau)\big] \leq - \frac{1-\rho}{2m^{\alpha/(1+\alpha)}}, \label{eq:iq-drift}
\end{equation}
whenever $L_{\alpha}(\mathbf{Q}(\tau))>B$.
\end{lemma}
\begin{proof}
For an IQ switch, and with a slight abuse of notation, we write $\bQ = [Q_{ij}]$. Recall that a  
schedule can serve $m$ queues simultaneously subject to the matching constraints. 
Because of this structural property, we claim that for any 
$\bQ = [Q_{ij}] \in \Zp^{m\times m}$, 
\begin{eqnarray}\label{eq:ratio-bound}
\sum_{i,j = 1}^m Q_{ij}^{\alpha+1}& \leq & m w_{\alpha}(\mathbf{Q})^{\frac{\alpha+1}{\alpha}}.
\end{eqnarray}
Equivalently, 
\begin{eqnarray}\label{eq:ratio-bound1}
w_{\alpha}(\mathbf{Q}(\tau))/||\mathbf{Q}(\tau)||_{\alpha+1}^{\alpha} & \geq & \frac{1}{m^{\alpha/(1+\alpha)}}.
\end{eqnarray}
By inspecting the proof of Theorem \ref{thm:sw-drift} for general switched
networks, we realize that the desired bound \eqref{eq:iq-drift} follows from \eqref{eq:ratio-bound1}. 
Therefore, to establish Lemma \ref{lem:iq-exp},
it is sufficient to verify \eqref{eq:ratio-bound}. 

Suppose that $\bsig^0 = [\sigma^0_{ij}]$ is a schedule with maximum $\alpha$-weight for a given $\bQ$. 
With some abuse of notation, let us denote $\sigma^0(i) = j$ if $\sigma^0_{ij} = 1$. Then,
$w_\alpha(\bQ) = \sum_{i=1}^m Q^\alpha_{i\sigma^0(i)}$. It can be shown that there exist 
$m-1$ other schedules (or matchings) $\bsig^1,\dots, \bsig^{m-1}$ so that all $m^2$ queues are served by the $m$ schedules $\bsig^0, \bsig^1,\dots, \bsig^{m-1}$. Since $w_\alpha(\bQ)$
is the maximum $\alpha$-weight, it follows that for any $k (0\leq k\leq m-1)$, 
\begin{equation}
w_\alpha(\bQ)\geq \sum_{i=1}^m Q^\alpha_{i\sigma^k(i)} \geq \left(\sum_{i=1}^m Q^{\alpha+1}_{i\sigma^k(i)}\right)^{\frac{\alpha}{1+\alpha}}, \label{x1}
\end{equation}
where the last inequality follows from the standard norm inequality 
$\|\bx\|_{\alpha} \geq \|\bx\|_{1+\alpha}$. Raising to the power $(1+\alpha)/\alpha$ on both sides of 
\eqref{x1}, then summing over all $k$, and using the property that the schedules $\bsig^{k}$
 cover all of the $m^2$ queues, we obtain
\begin{eqnarray}
m w_\alpha(\bQ)^{\frac{1+\alpha}{\alpha}} & \geq & \sum_{i,j=1}^m Q_{ij}^{1+\alpha}. 
\end{eqnarray}
This completes the verification of \eqref{eq:ratio-bound} and the proof of Lemma \ref{lem:iq-exp}.
\end{proof}

\paragraph{Lower Bound for IQ Switch} We now derive an exponential lower bound using 
the results of \cite{VL09} and assuming uniform arrival rates, i.e., $\blambda = [\rho/m]$ with
$\rho < 1$. In \cite{VL09}, the authors establish an LDP for the switched network model. They
show that under the MW-$\alpha$ policy, with arrival rates $\blambda$ satisfying $\rho(\blambda) < 1$, 
there exists $\theta_\alpha$ so that 
$$ \lim_{R\to\infty} \frac{1}{R}\log \mathbb{P}_{\bpi}\left(\|\mathbf{Q}(\tau)\|_{\alpha+1}>R\right) = -\theta_\alpha.$$ 
The tail exponent $\theta_\alpha$ is characterized as the solution of a variational problem:
\begin{equation}\label{eq:theta}
 \theta_{\alpha} = \inf \{ H(\tilde{\boldsymbol{\lambda}}|| \boldsymbol{\lambda})/r(\tilde{\boldsymbol{\lambda}}) :  \tilde{\boldsymbol{\lambda}} \in [0,1]^{m\times m}, \rho(\tilde{\boldsymbol{\lambda}})>1\},
 \end{equation}
where 
\begin{displaymath}
H(\tilde{\boldsymbol{\lambda}}|| \boldsymbol{\lambda}) = \sum_{i,j=1}^{m}\tilde{\lambda}_{ij}\log \left(\frac{\tilde{\lambda}_{ij}}{\lambda_{ij}}\right) + (1-\tilde{\lambda}_{ij})\log \left(\frac{1-\tilde{\lambda}_{ij}}{1-\lambda_{ij}}\right).
\end{displaymath}
Furthermore, $r(\tilde{\blambda})$ is the solution to the optimization problem:
\begin{center}
$\begin{array}{ll}
\textrm{minimize} & \|\mathbf{x}\|_{1+\alpha} \\
\textrm{subject to}     & r\in \mathbb{R}_+^M \\
                  & \mathbf{x} \geq \tilde{\boldsymbol{\lambda}}-\boldsymbol{\sigma}' \ \ \textrm{ for some } \boldsymbol{\sigma}' \in \Sigma,
\end{array}$
\end{center}
where $\Sigma$ is the convex hull of the set of feasible schedules $\mathcal{S}$. Clearly, 
an explicit formula for $\theta_\alpha$ in terms of $\rho(\blambda)$ and the switch
size $m$ seems impossible, even
 when $\blambda$ is uniform, i.e., $\blambda = [\rho/m]$. However, as we show next, it is possible
to obtain a useful lower bound. 

In order to obtain a lower bound on the large deviation probability, or equivalently, an upper bound
on $\theta_\alpha$, it is sufficient to restrict to symmetric overload arrival rates $\tilde{\blambda}$ 
in the optimization problem \eqref{eq:theta}. Thus, let us assume that 
$\tilde{\lambda}_{ij} = (1+\beps)/m$ for all $i,j$, where $\beps>0$. Then, it can be checked that
\begin{equation*}
r(\tilde{\boldsymbol{\lambda}}) = \left(m^2\left(\frac{1+\beps-1}{m}\right)^{\alpha+1}\right)^{\frac{1}{\alpha+1}} 
 = \beps m^{\frac{1-\alpha}{1+\alpha}}.
\end{equation*}
Therefore, the optimization in \eqref{eq:theta} reduces to minimizing
\begin{equation}
\frac{m^2}{\beps m^{\frac{1-\alpha}{1+\alpha}}} \left(\frac{1+\beps}{m}\log \frac{1+\beps}{\rho}+(1-\frac{1+\beps}{m})\log \frac{1-\frac{1+\beps}{m}}{1-\frac{\rho}{m}}\right),
\label{eq:ldp}
\end{equation}
over all $\beps>0$. Again, a closed form solution seems impossible, but we can look for
 an approximation. We are interested in comparing the bounds for large $m$ and $\rho$ near
$1$, and we will develop a good approximation in that regime.
We expect the optimizing value of $\beps$ in \eqref{eq:ldp} to be small. Therefore, 
we shall use the Taylor series expansion for $\log$ up to the first two terms, i.e. $\log (1+x)\approx x-x^2/2$. 
With these approximations, minimizing \eqref{eq:ldp} boils down to solving a quadratic equation. This
leads to an optimal solution $\beps^* \approx 1-\rho$. Indeed, if $\rho$ is near $1$, $\beps^*$ is
quite small, thus justifying our approximations.  Using $\beps^* \approx 1-\rho$, we obtain 
\begin{eqnarray}\label{eq:bound-theta}
\theta_\alpha & \leq & 2m^{2\alpha/(1+\alpha)}(1-\rho).
\end{eqnarray}
That is, for $\rho$ near $1$ and for $m$ large enough, we have 
\begin{align}\label{eq:iq-lb}
& \lim\inf_{R\to\infty} \frac{1}{R}\log \mathbb{P}_{\bpi}\left(\|\mathbf{Q}(\tau)\|_{\alpha+1}>R\right) \nonumber \\
& \geq - 2m^{2\alpha/(1+\alpha)}(1-\rho).
\end{align}
\paragraph{Comparison} Putting the bounds \eqref{eq:iq-ub} and \eqref{eq:iq-lb} together, we
obtain that 
\begin{align*}
-2m^{2\alpha/(1+\alpha)}(1-\rho)&\leq \liminf_{R\rightarrow\infty}\frac{1}{R}\log \mathbb{P}_{\bpi}\left(\|\mathbf{Q}(\tau)\|_{1+\alpha} >R\right)\\
&\leq \limsup_{R\rightarrow\infty}\frac{1}{R}\log \mathbb{P}_{\bpi}\left(\|\mathbf{Q}(\tau)\|_{1+\alpha}>R\right) \\
& \leq 
-\frac{1-\rho}{100} ~m^{-1-\frac{2}{\alpha+1}}.
\end{align*}
For any $\alpha$, ignoring small constants, the ratio between the two tail exponents is precisely $m^3$. 
From this, we see that the dependence of our upper bound exponent on the load $\rho$ is tight, when the system is heavily loaded. However, the dependence on the number of queues is not. 
\end{document}